%% file: main.tex
\documentclass[a4paper,UKenglish,cleveref, hyperref, thm-restate, numberwithinsect]{lipics-v2021}

\usepackage[disable]{todonotes}
\usepackage{newfile}
\usepackage{mathtools}
\usepackage{tikz}
\usetikzlibrary{positioning}

\input{hidelipics}

\hypersetup{%
colorlinks=true,
linkcolor=blue,
urlcolor=blue,
citecolor=blue
}
\nolinenumbers

\title{Reachability in Bidirected Pushdown VASS}

\author{Moses Ganardi}{Max Planck Institute for Software Systems (MPI-SWS), Germany}%
{ganardi@mpi-sws.org}%
{https://orcid.org/0000-0002-0775-7781}%
{}%

\author{Rupak Majumdar}{Max Planck Institute for Software Systems (MPI-SWS), Germany}%
{rupak@mpi-sws.org}%
{https://orcid.org/0000-0003-2136-0542}%
{}%

\author{Andreas Pavlogiannis}{Aarhus University, Aabogade 34, 8200, Aarhus, Denmark}%
{pavlogiannis@cs.au.dk}%
{https://orcid.org/0000-0002-8943-0722}%
{}%

\author{Lia Sch\"{u}tze}{Max Planck Institute for Software Systems (MPI-SWS), Germany}%
{lschuetze@mpi-sws.org}%
{https://orcid.org/0000-0003-4002-5491}%
{}%

\author{Georg Zetzsche}{Max Planck Institute for Software Systems (MPI-SWS), Germany}%
{georg@mpi-sws.org}%
{https://orcid.org/0000-0002-6421-4388}%
{}%



\authorrunning{M. Ganardi, R. Majumdar, A. Pavlogiannis, L. Sch\"{u}tze, G. Zetzsche}
\Copyright{Moses Ganardi, Rupak Majumdar, Andreas Pavlogiannis, Lia Sch\"{u}tze, Georg Zetzsche}

\input{macros}

\begin{document}

\maketitle

\input{abstract}

\section{Introduction}

\input{introduction}

\section{Preliminaries}

\input{results}

\section{Decidability}\label{sec:decidability}
\input{decidability}

\section{Ackermann upper bound}\label{sec:ackermann}

\input{ackermann}

\section{One-dimensional pushdown VASS}\label{sec:1d}
\input{one-dim}

\section{Tower lower bound}\label{sec:tower}

\input{tower}

\section{Conclusion}
\input{conclusion}

\label{beforebibliography}
\newoutputstream{pages}
\openoutputfile{main.pg}{pages}
\addtostream{pages}{\getpagerefnumber{beforebibliography}}
\closeoutputstream{pages}
\bibliography{bib}

\appendix

\section{Additional material for \cref{sec:ackermann}} \label{app:ackermann}

\input{app-ackermann}

\section{Additional material for \cref{sec:1d}} \label{app:1d}

\input{app-1d}

\section{Additional material for \cref{sec:tower}} \label{app:tower}

\input{app-tower}

\section{Failure of a $\PSPACE$ lower bound technique} \label{app:failure}
\input{app-pspace-lower-bound-failure}

\label{endofdocument}
\newoutputstream{pagestotal}
\openoutputfile{main.pgt}{pagestotal}
\addtostream{pagestotal}{\getpagerefnumber{endofdocument}}
\closeoutputstream{pagestotal}
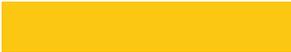
\end{document}

%% file: hidelipics.tex
\hideLIPIcs

%% file: macros.tex
\usepackage[vlined,ruled,linesnumbered]{algorithm2e}
\usepackage{bm}

\newcommand{\ComplexityFont}[1]{\ensuremath{{\small\mathsf{#1}}}}
\newcommand{\TOWER}{\ComplexityFont{TOWER}}
\newcommand{\PR}{\text{\ComplexityFont{PRIMITIVE}-\ComplexityFont{RECURSIVE}}}
\newcommand{\ACK}{\ComplexityFont{ACKERMANN}}
\newcommand{\PSPACE}{\ComplexityFont{PSPACE}}
\newcommand{\EXPSPACE}{\ComplexityFont{EXPSPACE}}
\newcommand{\NP}{\ComplexityFont{NP}}
\newcommand{\PTIME}{\ComplexityFont{P}}

\newcommand{\Cong}{\mathsf{Cong}}

\newcommand{\F}{\ComplexityFont{F}}

\renewcommand{\P}{\mathcal{P}}
\newcommand{\Q}{\mathcal{Q}}
\renewcommand{\S}{\mathcal{S}}
\newcommand{\Z}{\mathbb{Z}}
\newcommand{\N}{\mathbb{N}}

\newcommand{\gz}[1]{}

\newcommand{\Dominated}{\leq}
\newcommand{\StrictDominated}{<}
\newcommand{\Weight}{w}
\newcommand{\MaxStackHeight}{\mathit{MaxSH}}
\newcommand{\PathSum}{\Weight}
\newcommand{\PathMin}{\mathit{m}}
\newcommand{\Paths}[2]{\{ #1\rightsquigarrow #2 \}}
\newcommand{\DFinite}{\gamma}
\newcommand{\DOmega}{\delta}
\newcommand{\Extend}{\odot}
\newcommand{\DS}{\mathit{DS}}
\newcommand{\Some}{\_}
\newcommand{\Concat}{\circ}
\newcommand{\tuple}[1]{\langle #1 \rangle}
\newcommand{\ov}{\overline}
\newcommand{\Reverse}[1]{\ov{#1}}
\newcommand{\StackAlphabetBot}{\Gamma_{\bot}}
\newcommand{\Project}{\downharpoonright}
\newcommand{\DTo}[1]{\xrightarrow{#1}}
\newcommand{\Zstep}{\hookrightarrow}

%% file: abstract.tex
\begin{abstract}

A pushdown vector addition system with states (PVASS) extends the model of vector addition systems with a pushdown store.
A PVASS is said to be \emph{bidirected} if every transition (pushing/popping a symbol or modifying a counter)
has an accompanying opposite transition that reverses the effect.
Bidirectedness arises naturally in many models; it can also be seen as a overapproximation of reachability.
We show that the reachability problem for \emph{bidirected} PVASS is decidable in Ackermann
time and primitive recursive for any fixed dimension.
For the special case of one-dimensional bidirected PVASS, we show reachability is in $\mathsf{PSPACE}$, and in fact in 
polynomial time if the stack is polynomially bounded. 
Our results are in contrast to the \emph{directed} setting, where decidability of reachability is a long-standing open problem already
for one dimensional PVASS, and there is a $\mathsf{PSPACE}$-lower bound already for one-dimensional PVASS with bounded stack.

The reachability relation in the bidirected (stateless) case is a congruence over $\mathbb{N}^d$. 
Our upper bounds exploit saturation techniques over congruences.
In particular, we show novel elementary-time constructions of semilinear representations of congruences generated by finitely many vector pairs.
In the case of one-dimensional PVASS, we employ a saturation procedure over bounded-size counters.

We complement our upper bound with a $\mathsf{TOWER}$-hardness result for arbitrary dimension and 
$k$-$\mathsf{EXPSPACE}$ hardness in dimension $2k+6$ using a technique by Lazi\'{c} and Totzke  
to implement iterative exponentiations. 

\end{abstract}

%% file: introduction.tex
The reachability problem for infinite-state systems is one of the most basic
and well-studied tasks in verification. Given an infinite-state system and two
configurations $c_1$ and $c_2$ in the system, it asks: Is there a run from
$c_1$ to $c_2$?  
\emph{Pushdown systems} (PDS) and
\emph{vector addition systems with states} (VASS) are prominent models for which the reachability
problem has been studied extensively. 
Each of them features a finite set of control states and a storage mechanism that holds an unbounded
amount of information. In a PDS, there is a stack where we can push
and pop letters. In a VASS, there is a set of counters which can be
\emph{incremented} and \emph{decremented}, but not tested for zero.
Reachability in both models is understood in isolation
\cite{DBLP:conf/concur/BouajjaniEM97,CzerwinskiLLLM19,DBLP:conf/focs/Leroux21,DBLP:conf/focs/CzerwinskiO21,LerouxS19}, 
but the reachability problem for their \emph{combination} is a long-standing open problem.

\subparagraph{Pushdown VASS} 
A \emph{pushdown VASS} (PVASS) combines PDS and VASS.
A PVASS consists of finitely many control states and has access to
both a pushdown stack (as in PDS) and counters (as in VASS). 
A PVASS is \emph{$d$-dimensional} if it has $d$ counters. A PVASS is a natural
combination of the simple building blocks of PDS and VASS. The reachability problem for PVASS
has remained a long-standing open problem~\cite{LerouxST15,DBLP:conf/rp/SchmitzZ19,EnglertHLLLS21},
even if we combine a pushdown with a single counter.  

\subparagraph{Bidirectedness}
A step toward deciding reachability is to first study natural relaxations of the reachability relation.
A relaxation that has recently attracted attention is \emph{bidirectedness}. 
Bidirectedness assumes that for each transition from state $p$
to $q$ in our infinite-state system, there exists a transition from $q$ to $p$
with opposite effect. For example, in bidirected pushdown systems, for each
transition from $p$ to $q$ pushing $\gamma$ on the stack, there is a transition
from $q$ to $p$ that pops $\gamma$. Likewise, in bidirected VASS, if there is a transition from $p$ to $q$
that adds some vector $\bm{v}\in\Z^d$ to counters, then there is a transition
from $q$ to $p$ adding $-\bm{v}$. It turns out that several tasks in program
analysis can be formulated or practically approximated as reachability in
bidirected pushdown systems~\cite{ChatterjeeCP2018,DBLP:conf/pldi/ZhangLYS13}
or bidirected multi-pushdown
systems~\cite{DBLP:conf/ecoop/XuRS09,DBLP:conf/issta/YanXR11,DBLP:conf/popl/ZhangS17,DBLP:conf/pldi/LiZR20,DBLP:journals/pacmpl/LiZR21,DBLP:journals/pacmpl/KP22}.
Bidirected systems have also been considered in algorithmic group theory as an
algorithmic framework to provide simple algorithms for the membership problem
in subgroups~\cite{lohrey2010automata}.

Reachability in bidirected systems is usually considerably more efficient than
in the general case. In bidirected pushdown systems, reachability can be solved
in almost linear time~\cite{ChatterjeeCP2018} whereas a truly subcubic
algorithm for the general case is a long-standing open
problem~\cite{HeintzeMcAllester,Chaudhuri}.  
Reachability in bidirected VASS is equivalent to the
uniform word problem in finitely presented commutative semigroups, which is
$\EXPSPACE$-complete~\cite{MM82}.  
A separate polynomial time algorithm for bidirected two-dimensional VASS was given
in~\cite{DBLP:journals/pacmpl/LiZR21}. Moreover, recent results on reachability
in bidirected valence systems shows complexity drops across a large variety of
infinite-state systems~\cite{GanardiMajumdarZetzsche2022a}: For almost
every class of systems studied in~\cite{GanardiMajumdarZetzsche2022a}, the
complexity of bidirected reachability is lower than in the general case (the
only exception being pushdown systems, where the complexity is
$\PTIME$-complete in both settings). For example, reachability in bidirected
$\Z$-VASS, and even in bidirected $\Z$-PVASS, is in
$\PTIME$~\cite{GanardiMajumdarZetzsche2022a}.

However, little is known about bidirected PVASS.  They have recently been
studied in~\cite{DBLP:journals/pacmpl/KP22}, where decidability of reachability
in {\em dimension one is shown}.  
However, as in the non-bidirected case, decidability of reachability in bidirected PVASS is hitherto not known.

\subparagraph{Contributions}

We show that in bidirected PVASS (of arbitrary dimension), reachability is decidable.
Moreover, we provide an Ackermann complexity upper bound, and show that in any
fixed dimension, reachability is primitive recursive.

\begin{theorem}
\label{thm:bidirected-pvass}
Reachability in bidirected pushdown VASS is in $\ACK$,
and primitive recursive (in $\F_{4d+11}$) if the dimension $d$ is fixed.
\end{theorem}

Here, $(\F_\alpha)_\alpha$ is an ordinal-indexed hierarchy of fast-growing complexity classes \cite{Schmitz16},
including $\F_3 = \TOWER$ and $\F_\omega = \ACK$. 
The formal definition of the hierarchy can be found in \cref{sec:bad-seq}.
A recurring theme in our upper bounds is that saturation techniques, the
standard method to analyze pushdown systems, combine surprisingly well with
counters in the bidirected setting. 
Saturation is used in each of our upper bounds. In~\cref{sec:decidability}, we
begin the exposition with a short, self-contained proof that reachability is
decidable in bidirected PVASS. It shows that non-reachability is always
certified by an inductive invariant of a particular saturation procedure.
In~\cref{sec:ackermann}, we show the Ackermann upper bound. Here, we saturate a
congruence relation that encodes the reachability relation. The upper bound
relies on two key ingredients. First, we use results about Gr\"{o}bner bases of
polynomial ideals to show that in elementary time, one can construct a
Presburger formula for the congruence generated by finitely many vector pairs.
This construction serves as one step in the saturation. To show termination in
Ackermannian time, we rely on a technique from~\cite{FigueiraFSS11} to bound
the length of strictly ascending chains of upward closed sets of vectors. Here,
the difficulty is to transfer this bound from chains of upward closed sets to
chains of congruences.

In \cref{sec:1d} we present a $\PSPACE$ algorithm for bidirected PVASS in dimension one.

\begin{theorem}\label{thm:bidirected-pvass-1d}
Reachability in 1-dimensional bidirected pushdown VASS is in $\PSPACE$.
\end{theorem}

Here, we rely on an observation
from~\cite{DBLP:journals/pacmpl/KP22} that reachability in bidirected one-dimensional
PVASS reduces to (i)~coverability in bidirected one-dimensional PVASS and
(ii)~reachability in one-dimensional bidirected $\Z$-PVASS. Since (ii)~is known
to be in $\PTIME$~\cite{GanardiMajumdarZetzsche2022a}, we show that (i)~can be done in $\PSPACE$. For this, we use saturation to compute, for
each state pair $(p,q)$, three bounds on counter values that determine whether
coverability holds. We show that these bounds have at most exponential absolute value,
which yields a $\PSPACE$ procedure.

Finally in~\cref{sec:tower}, we show that reachability in bidirected PVASS is
$\TOWER$-hard.  For this, we adapt a technique from~\cite{Lazic2017} that shows
a $\TOWER$ lower bound for general PVASS.

\begin{theorem}
	\label{thm:lower-bound}
	Reachability in bidirected PVASS is $\TOWER$-hard, and $k$-\ComplexityFont{EXPSPACE}-hard
	in dimension $2k+6$.
\end{theorem}

\subparagraph{Related work}

The model of pushdown VASS is surrounded
by extensions and restrictions of the storage mechanism for which decidability
is understood, the most prominent being the recent Ackermann-completeness for reachability in VASS
\cite{CzerwinskiLLLM19,DBLP:conf/focs/Leroux21,DBLP:conf/focs/CzerwinskiO21,LerouxS19}. 
If instead of the stack, we have a counter with zero tests, then
reachability is still
decidable~\cite{DBLP:journals/entcs/Reinhardt08,DBLP:journals/corr/abs-1205-4458}.
Here, decidability even holds if we have a zero-testable counter and one
additional counter that can be
reset~\cite{DBLP:conf/stacs/FinkelS00,DBLP:conf/fsttcs/FinkelLS18}. 
Furthermore, the extension of VASS by \emph{nested zero tests}, where for each
$i\in\{1,\ldots,d\}$, we have an instruction that tests all counters
$1,\ldots,i$ for zero simultaneously, also allows deciding
reachability~\cite{DBLP:journals/entcs/Reinhardt08,bonnet2013theory} and can be
seen as a special case of pushdown VASS~\cite{DBLP:conf/fsttcs/AtigG11}.
Another decidability result concerns the \emph{coverability problem}: Here, we
are given a configuration $c_1$ and a control state $q$ and want to know
whether from $c_1$, one can reach some configuration in control state $q$.  It
is known that the reachability problem for $d$-dimensional PVASS reduces to
coverability in $(d+1)$-dimensional PVASS, and that coverability in
$1$-dimensional PVASS is decidable~\cite{LerouxST15}. According to
\cite{EnglertHLLLS21}, the latter problem is $\PSPACE$-hard
and in $\EXPSPACE$.
Furthermore, if the counters in a PVASS are allowed to go negative during a run,
then we speak of an \emph{integer PVASS} ($\Z$-PVASS). For these, reachability 
is known to be decidable~\cite{DBLP:journals/jcss/HarjuIKS02} and $\NP$-complete~\cite{DBLP:conf/cav/HagueL11}.
However, if we extend the model of PVASS by allowing resets on the counters,
then even coverability is undecidable in dimension one~\cite{DBLP:conf/rp/SchmitzZ19}.

For VASS, several generalizations of bidirectedness have been studied.  It is $\EXPSPACE$-complete whether given two configurations are
mutually reachable~\cite{DBLP:journals/corr/abs-1301-4874}. Moreover,
if two configurations are mutually reachable, then their distance is at most
doubly exponential (linear for fixed dimension) in their
size~\cite{DBLP:conf/fsttcs/Leroux19}. Furthermore, for cyclic VASS (where each
transition can be reversed by some execution), it is known that the
reachability set has a semilinear representation of at most exponential
size~\cite{DBLP:journals/entcs/BouzianeF97}. Let us note that in the VASS/Petri
net literature, sometimes~\cite{DBLP:journals/entcs/BouzianeF97} (but not
entirely consistently~\cite{DBLP:journals/corr/abs-1301-4874}) the term
\emph{reversible} is used to mean bidirected. However, this clashes with the
reversibility notion in dynamical systems~\cite{DBLP:journals/ngc/Kari18}.

%% file: results.tex
\subparagraph{Vectors and semilinear sets}
We denote integer vectors by bold letters $\bm{x}$.
The maximum norm of $\bm{x}$ is denoted by $\|\bm{x}\|$.
The $i$-th unit vector is denoted by $\bm{e}_i$.
The componentwise order $\le$ on $\N^d$ is a {\em well-quasi order (wqo)},
i.e. for any infinite sequence $\bm{x}_1, \bm{x}_2, \dots$ over $\N^d$ there exist $i < j$ with $\bm{x}_i \le \bm{x}_j$. We write $\bm{x}<\bm{y}$ if $\bm{x}\le\bm{y}$ and $\bm{x}\ne\bm{y}$.
This implies that the set $\min(X)$ of minimal elements in any $X \subseteq \N^d$ is finite.
We denote by $X {\uparrow} = \{ \bm{y} \in \N^k \mid \exists \bm{x} \in X \colon \bm{x} \le \bm{y} \}$
the {\em upwards closure} of $X$.
We also write $\bm{x} {\uparrow}$ for $\{ \bm{x} \} {\uparrow}$.
A {\em congruence} on a commutative monoid $(M,+)$, for example $M = \N^d$,
is an equivalence relation $\Q \subseteq M \times M$ 
where $(a,b) \in \Q$ implies $(a+c,b+c) \in \Q$
for all $a,b,c \in M$.
We also write $a \sim_\Q b$ instead of $(a,b) \in \Q$.

For $X \subseteq \N^k$ we denote by $X^*$ the {\em submonoid} generated by $X$.
A set $L \subseteq \N^k$ is {\em linear} if it is of the form $L = \bm{b} + P^*$
for some base vector $\bm{b} \in \N^k$ and some finite set $P \subseteq \N^k$ of period vectors.
Finite unions of linear sets are called {\em semilinear}.
It is well-known that a set is semilinear if and only if it is definable in {\em Presburger arithmetic},
i.e.\ first-order logic over $(\N,+,\le,0,1)$.
Furthermore, one can effectively convert between these formats in elementary time:
While defining semilinear sets in Presburger arithmetic is straightforward,
for the converse we can use Cooper's quantifier elimination \cite{Cooper72} running in triply exponential time \cite{Oppen78},
see also \cite{Haase18} for an excellent overview.
We will confuse a semilinear $S$ with its representation,
which is either a list of base and period vectors for each linear set or a defining Presburger formula,
and denote by $\|S\|$ the size of its representation.

\subparagraph{Pushdown VASS}

A {\em $d$-dimensional pushdown VASS (PVASS)} is a tuple $\P = (Q,\Gamma,T)$
where $Q$ is a finite set of states, $\Gamma$ is a finite stack alphabet,
and $T \subseteq Q \times \Z^d \times \mathsf{Op}(\Gamma) \times Q$ is a finite set of transitions.
Here $\mathsf{Op}(\Gamma) = \{ a, \bar a \mid a \in \Gamma \} \cup \{\varepsilon\}$ is the set of
operations on the stack.
A configuration over $\P$ is a tuple $(q,\bm{x},s) \in Q \times \N^d \times \Gamma^*$.
The {\em one-step relation} $\to$ is the smallest binary relation on configurations such that
for all $(p,\bm{v},\alpha,q) \in T$ and $\bm{x} \in \N^d$ with $\bm{x}+\bm{v} \ge \bm{0}$ we have:
(i) If $\alpha \in \Gamma \cup \{\varepsilon\}$ then $(p,\bm{x},s) \to (q,\bm{x}+\bm{v},s \alpha)$
(ii) if $\alpha = \bar a$ then $(p,\bm{x},sa) \to (q,\bm{x}+\bm{v},s)$.
Its transitive-reflexive closure is denoted by $\xrightarrow{*}$.
We say that $\P$ is {\em bidirected} if $(p,\bm{v},\alpha,q) \in T$ implies $(q,-\bm{v},\bar \alpha,p) \in T$
where we set $\bar{\bar a} = a$ for $a \in \Gamma$ and $\bar \varepsilon = \varepsilon$.
The {\em reachability problem} for bidirected PVASS asks:
Given a bidirected PVASS $\P$ and two states $s,t$, does $(s,\bm{0},\varepsilon) \xrightarrow{*} (t,\bm{0},\varepsilon)$ hold?

The counter updates $\bm{u}$ in a PVASS transition $(p,\bm{u},q)$ can be given in either unary or binary encoding
since there are logspace translations in both directions:
To add a binary encoded number $u$ to a counter we push the binary notation of $u$ to the stack,
and repeatedly decrement the stack counter while incrementing $u$.
Since this computation is deterministic, the simulation also works for bidirected PVASS.

For the Ackermann upper bound it is convenient to use pushdown VASS with a single state.
A {\em pushdown VAS (PVAS)} $\P = (\Gamma,T)$ in dimension $d$ consists of a finite stack alphabet $\Gamma$
and a finite set of transitions $T \subseteq \N^d \times \N^d \times (\Gamma \cup \bar \Gamma \cup \{\varepsilon\})$.
Here, a configuration is a pair $(\bm{x},s) \in \N^d \times \Gamma^*$.
The effect of a transition $(\bm{u},\bm{v},\alpha)$ is subtracting $\bm{u}$ from the $d$ counters,
assuming that the counters stay non-negative, and then adding $\bm{v}$.
A PVAS $\P$ is bidirected if $(\bm{u},\bm{v},a) \in T$ implies $(\bm{v},\bm{u},\bar a) \in T$.
A bidirected PVASS in dimension $d$ can be simulated by a bidirected PVAS in dimension $d+2$
where the two additional counters add up to the number of states and specify the current state.
Hence, one can reduce the reachability problem for bidirected PVASS to the reachability problem for bidirected PVAS:
Given a bidirected PVAS $\P$ and two vectors $\bm{s},\bm{t} \in \N^d$,
does $(\bm{s},\varepsilon) \xrightarrow{*} (\bm{t},\varepsilon)$ hold?

%% file: decidability.tex
In this section, we present a simple and self-contained proof that reachability
is decidable in bidirected PVASS. Consider the reachability
relation between configurations with empty stack.
For any states $p,q$, define the set $R_{p,q}\subseteq\N^d\times\N^d$ with
\[ R_{p,q}=\{(\bm{u},\bm{v}) \mid \bm{u},\bm{v}\in\N^d,~(p,\bm{u},\varepsilon)\xrightarrow{*}(q,\bm{v},\varepsilon) \}. \]
We will prove that each $R_{p,q}$ is semilinear, for which we rely on the fact that these sets are slices.
A \emph{slice} is a subset $S\subseteq\N^k$ such that if
$\bm{u},\bm{u}+\bm{v},\bm{u}+\bm{w}\in S$ for some $\bm{u},\bm{v},\bm{w}\in\N^k$, then $\bm{u}+\bm{v}+\bm{w}\in S$.
Observe that each $R_{p,q}\subseteq\N^{2d}$ is a slice. This is because if $(\bm{u},\bm{v}), (\bm{u}+\bm{u}_1,\bm{v}+\bm{v}_1), (\bm{u}+\bm{u}_2,\bm{v}+\bm{v}_2)\in R_{p,q}$, then there is a run
\[ (p,\bm{u}+\bm{u}_1+\bm{u}_2,\varepsilon)\xrightarrow{*}(q,\bm{v}+\bm{v}_1+\bm{u}_2,\varepsilon)\xrightarrow{*}(p,\bm{u}+\bm{v}_1+\bm{u}_2,\varepsilon)\xrightarrow{*}(q,\bm{v}+\bm{v}_1+\bm{v}_2,\varepsilon), \]
where the middle part exists due to bidirectedness.
Thus, the pair $(\bm{u}+\bm{u}_1+\bm{u}_2,\bm{v}+\bm{v}_1+\bm{v}_2)$ belongs to $R_{p,q}$.
The following was first shown in \cite[Proposition 7.3]{EilenbergSchutzenberger1969}.
\begin{theorem}[Eilenberg \& Sch\"{u}tzenberger 1969]\label{slice-semilinear}
	Every slice is semilinear.
\end{theorem}
This seems to be stronger than the somewhat better-known fact that each
congruence on $\N^d$ is semilinear: Observe that every congruence on
$\N^d$, seen as a subset of $\N^{2d}$, is a slice. In the case of
congruences, a relatively simple proof was obtained by
Hirshfeld~\cite{hirshfeld1994congruences}. We present a proof of
\cref{slice-semilinear} that combines ideas from both
\cite{EilenbergSchutzenberger1969} and \cite{hirshfeld1994congruences} and is
(in our opinion) simpler than each. 

\newcommand{\ucl}[1]{#1\mathop{\uparrow}}
For a set $X\subseteq\N^k$, let $\min X$ be the set of minimal elements of $X$,
with respect to the usual component-wise ordering $\le$ on $\N^k$. Since this
ordering is a well-quasi ordering, $\min X$ is finite for every set $X$.
Suppose $S\subseteq\N^k$ is a slice. For each $\bm{u}\in S$, let
$S-\bm{u}:=\{\bm{v}\in\N^k \mid \bm{u}+\bm{v}\in S\}$.
Then $\bm{u} \le \bm{v}$ implies $S-\bm{u} \subseteq S-\bm{v}$.
Consider for each $\bm{u}\in S$ the submonoid
\[ M_{\bm{u}}=(\min (S-\bm{u}\setminus\{\bm{0}\}))^*. \]
In other words, $M_{\bm{u}}$ is the submonoid of $\N^k$ generated by the
non-zero minimal elements of $S-\bm{u}$. Note that for $\bm{u},\bm{v}\in S$, we
have $M_{\bm{u}}=M_{\bm{v}}$ if and only if
$\ucl{(S-\bm{u}\setminus\{\bm{0}\})}=\ucl{(S-\bm{v}\setminus\{\bm{0}\})}$.
Since $S$ is a slice, we have $\bm{u}+M_{\bm{u}}\subseteq S$ for every $\bm{u}
\in S$.  
Since $\bm{u}\in\bm{u}+M_{\bm{u}}$, we trivially have
\[ S=\bigcup_{\bm{u}\in S} \bm{u}+M_{\bm{u}}. \]
Since each $\bm{u}+M_{\bm{u}}$ is semilinear, it suffices to show that $S$ is
covered by finitely many sets $\bm{u}+M_{\bm{u}}$. We first observe that if
$\bm{u}\le\bm{v}$ and $M_{\bm{u}}=M_{\bm{v}}$, then $\bm{u}+M_{\bm{u}}$ already
covers $\bm{v}+M_{\bm{v}}$.
\begin{lemma}\label{smaller-cover}
	Let $\bm{u},\bm{v}\in S$. If $\bm{u}\le\bm{v}$ and $M_{\bm{u}}=M_{\bm{v}}$, then $\bm{v}+M_{\bm{v}}\subseteq\bm{u}+M_{\bm{u}}$.
\end{lemma}
\begin{proof}
	We will use the following claim: For every $\bm{w}\in S$ with
	$\bm{u}\le\bm{w}\le\bm{v}$, we have $M_{\bm{u}}=M_{\bm{w}}=M_{\bm{v}}$.
	Indeed, since $M_{\bm{u}}=M_{\bm{v}}$, we have
	$\min(S-\bm{u}\setminus\{\bm{0}\})=\min(S-\bm{v}\setminus\{\bm{0}\})$.  Moreover, since $S$ is a slice, we
	have $S-\bm{u}\subseteq S-\bm{w}\subseteq S-\bm{v}$. Therefore,
	$\min(S-\bm{w}\setminus\{\bm{0}\})$ coincides with $\min(S-\bm{u}\setminus\{\bm{0}\})$ and $\min(S-\bm{v}\setminus\{\bm{0}\})$,
	which implies $M_{\bm{w}}=M_{\bm{u}}=M_{\bm{v}}$.

	Let us prove the lemma.   We proceed by induction on $\|\bm{v}-\bm{u}\|$.
	If $\bm{u}=\bm{v}$, then we are done. Otherwise, there exists an
	$\bm{m}\in \min(S-\bm{u}\setminus\{\bm{0}\})$ such that
	$\bm{u}+\bm{m}\le\bm{v}$.  By our claim, we have
	$M_{\bm{u}}=M_{\bm{u}+\bm{m}}=M_{\bm{v}}$. Therefore, induction implies
	$\bm{v}\in \bm{u}+\bm{m}+M_{\bm{u}+\bm{m}}$. But since $\bm{m}\in
	M_{\bm{u}}$ and $M_{\bm{u}+\bm{m}}\subseteq M_{\bm{u}}$, this implies
	$\bm{v}\in \bm{u}+M_{\bm{u}}$.
\end{proof}

The following implies semilinearity of $S$.
\begin{lemma}
	There is a finite set $F\subseteq S$ such that $S=\bigcup_{\bm{u}\in F} \bm{u}+M_{\bm{u}}$.
\end{lemma}
\begin{proof}
	Suppose not. Then there is an infinite sequence
	$\bm{u}_1,\bm{u}_2,\ldots\in S$ such that each set $\bm{u}_i+M_{\bm{u}_i}$
	contributes a new element.
	By Dickson's lemma
	$\bm{u}_1,\bm{u}_2,\ldots$ contains a subsequence $\bm{v}_1, \bm{v}_2, \ldots$ with $\bm{v}_i \le \bm{v}_{i+1}$ for all $i \ge 1$.
	Since then $S-\bm{v}_1\subseteq S-\bm{v}_2\subseteq\cdots$, the sequence
	$\ucl{(S-\bm{v}_1 \setminus\{\bm{0}\} )}\subseteq \ucl{(S-\bm{v}_2 \setminus\{\bm{0}\})}\subseteq\cdots$ becomes stationary,
	again by Dickson's lemma,
	and therefore also the sequence $M_{\bm{v}_1},M_{\bm{v}_2},\ldots$. By
	\cref{smaller-cover}, this means that only finitely many terms in the
	sequence $\bm{v}_1+M_{\bm{v}_1},\bm{v}_2+M_{\bm{v}_2},\ldots$
	contribute new elements, a contradiction.
\end{proof}

\subparagraph{Saturation invariants} We have seen that the reachability
relations $R_{p,q}$ are all semilinear. However, since the semilinearity proof
is non-constructive, this does not explain how to decide reachability.
Nevertheless, we shall use semilinearity to show that in case of
non-reachability, there exists a certificate.  This yields a decision procedure
consisting of two semi-algorithms in the style of Leroux's algorithm for
reachability in VASS~\cite{DBLP:conf/popl/Leroux11}: One semi-algorithm
enumerates potential runs, and one enumerates potential certificates for
non-reachability.

We assume that we are given a
bidirected $d$-dimensional PVASS with state set $Q$ and stack alphabet
$\Gamma$.  We may assume that all transitions are of the form
$p\xrightarrow{\gamma}q$ or $p\xrightarrow{\bar{\gamma}}q$ for
$\gamma\in\Gamma$ or $p\xrightarrow{\bm{v}}q$ for $\bm{v}\in\Z^d$.  
Our certificates for non-reachability will be in the form of what we call
saturation invariants. Imagine a (non-terminating) naive saturation algorithm
that attempts to compute the sets $R_{p,q}$ by adding vector pairs one-by-one
to finite sets $F_{p,q}$.  It would start with $F_{p,q}=\emptyset$ and then add
pairs: For each transition $p\xrightarrow{\bm{v}}q$ and each vector
$\bm{u}\in\N^d$ with $\bm{u}+\bm{v}\in\N^d$, it would add the pair
$(\bm{u},\bm{u}+\bm{v})$ to $F_{p,q}$. Moreover, if $(\bm{u},\bm{v})\in
F_{p,q}$ and $(\bm{v},\bm{w})\in F_{q,r}$, it would add $(\bm{u},\bm{w})$ to
$F_{p,r}$.  Finally, if there are transitions $p\xrightarrow{\gamma}p'$ and
$q'\xrightarrow{\bar{\gamma}}q$ and there is a $(\bm{u},\bm{v})\in F_{p',q'}$,
then it would add $(\bm{u},\bm{v})$ to $F_{p,q}$.

Intuitively, a saturation invariant is a forward inductive invariant of this
naive saturation algorithm. Let us make this precise.  For subsets
$R_1,R_2\subseteq\N^d\times\N^d$, we define
\[ R_1\circ R_2 = \{(\bm{u},\bm{w})\in\N^d\times\N^d \mid \exists \bm{v}\in\N^d\colon (\bm{u},\bm{v})\in R_1,~(\bm{v},\bm{w})\in R_2\}. \]
A \emph{saturation invariant} consists of a family $(I_{p,q})_{(p,q)\in Q^2}$
of sets $I_{p,q}\subseteq\N^d\times\N^d$ for which
\begin{enumerate}
	\item For each transition $p\xrightarrow{\bm{v}} q$, $\bm{v}\in\Z^d$, each $\bm{u}\in\N^d$ with $\bm{u}+\bm{v}\in\N^d$, we have $(\bm{u}, \bm{u}+\bm{v})\in I_{p,q}$.
	\item For each $p,q,r\in Q$, we have $I_{p,q}\circ I_{q,r}\subseteq I_{p,r}$.
	\item For each $p,p',q,q'\in Q$ for which there are transitions $p\xrightarrow{\gamma}p'$, $q'\xrightarrow{\bar{\gamma}} q$ for some $\gamma\in\Gamma$,
		we have $I_{p',q'}\subseteq I_{p,q}$.
\end{enumerate}
There is a natural ordering of such families $(I_{p,q})_{(p,q)\in Q^2}$ 
defined by inclusion: We write $(I_{p,q})_{(p,q)\in Q^2}\le (J_{p,q})_{(p,q)\in Q^2}$,
if $I_{p,q}\subseteq J_{p,q}$ for each $p,q\in Q$. In this sense, we can speak
of a smallest saturation invariant.
\begin{lemma}\label{decidability-smallest-saturation-invariant}
	The family $(R_{p,q})_{(p,q)\in Q^2}$ is the smallest saturation invariant.
\end{lemma}
\begin{proof}
	By induction on the length of a run, it follows that $(R_{p,q})_{(p,q)\in Q^2}$
	is included in every saturation invariant. Moreover, $(R_{p,q})_{(p,q)\in Q^2}$
	is clearly a saturation invariant itself.
\end{proof}

Our certificates will consist of saturation invariants defined in Presburger
arithmetic.  A family $(I_{p,q})_{(p,q)\in Q^2}$ is \emph{Presburger-definable}
if for each $(p,q)\in Q^2$, the set $I_{p,q}$ is semilinear. According to
\cref{slice-semilinear}, the family $(R_{p,q})_{(p,q)\in Q^2}$ is
Presburger-definable.  Therefore, the following is a direct consequence of
\cref{decidability-smallest-saturation-invariant}.
\begin{theorem}\label{thm-invariants}
	For each $s,t\in Q$, we have $(\bm{0},\bm{0})\notin R_{s,t}$ if and
	only if there exists a Presburger-definable saturation invariant
	$(I_{p,q})_{(p,q)\in Q^2}$ such that $(\bm{0},\bm{0})\notin I_{s,t}$.
\end{theorem}

This yields our algorithm: One semi-algorithm enumerates transition sequences
and terminates if one of them is a run witnessing
$(s,\bm{0},\varepsilon)\xrightarrow{*}(t,\bm{0},\varepsilon)$.  The other
semi-algorithm enumerates Presburger-definable families $(I_{p,q})_{(p,q)\in
Q^2}$ in the form of Presburger formulas. Using Presburger arithmetic, it is
then easy to check whether (i)~$(I_{p,q})_{(p,q)\in Q^2}$ is a saturation
invariant and (ii)~$(\bm{0},\bm{0})\notin I_{s,t}$. If a saturation invariant
is found, the semi-algorithm reports non-reachability. By
\cref{thm-invariants}, one of the two semi-algorithms must terminate.

%% file: ackermann.tex
In this section, we show that reachability in bidirected PVASS is solvable in
Ackermann time in the general case and in primitive recursive complexity in
every fixed dimension.

One way to avoid enumeration in the algorithm of \cref{sec:decidability} would
be to start with the semilinear one-step relation described in the first
condition of saturation invariants, and then to enlarge it according to the
second and third condition. Moreover, one could take the slice closure (the
smallest slice that includes the current set) after each enlargement. Since
slices satisfy an ascending chain condition~\cite[Corollary
12.3]{EilenbergSchutzenberger1969}, this would ensure termination.  In fact,
computing the slice closure of a semilinear set is possible with an algorithm
by Grabowski~\cite{grabowski1981}. Unfortunately, the latter is itself based on
enumeration and we are not aware of any complexity bounds for computing slice
closures.  Therefore, we use an analogous algorithm that uses congruences
instead of slices. Since congruences can be encoded in polynomial ideals, we
can tap into the rich toolbox of Gröbner bases to compute the congruence generated
by a semilinear set.

\subsection{The saturation algorithm}

In the following we will work with pushdown VAS instead of pushdown VASS.
Our decision procedure for bidirected reachability relies on the crucial fact
that the reachability relation
$R_\P = \{ (\bm{s},\bm{t}) \in \N^d \times \N^d \mid (\bm{s},\varepsilon) \xrightarrow{*} (\bm{t},\varepsilon) \}$
of a bidirected pushdown VAS $\P$ is a congruence:
It is always reflexive, transitive and additive, even for directed pushdown VAS,
and symmetric for bidirected systems.
Therefore, whenever we have found an underapproximation $R \subseteq R_\P$
we can replace $R$ by the smallest congruence containing $R$.
The smallest congruence containing a set $R \subseteq \N^d \times \N^d$ is denoted by $\Cong(R)$.
We also say that $R$ is a {\em basis of} (or {\em generates}) $\Cong(R)$.
Recall that every congruence on $\N^d$ is a slice.
Therefore, congruences are semilinear and ascending chains of congruences stabilize.

\begin{algorithm}[t]
\KwData{Bidirected $d$-dim. PVAS $\P = (\Gamma,T)$}
$R_0 := \Cong(\{ (\bm{u},\bm{v}) \mid (\bm{u},\bm{v},\varepsilon) \in T \})$\;
\For{$i = 1, 2, \dots$}{
  $R_i \leftarrow R_{i-1}$\;
  \For{$(\bm{u},\bm{u}',a) \in T$ and $(\bm{v}',\bm{v},\bar a) \in T$}{
    $R_i \leftarrow R_i \cup \{ (\bm{x}+\bm{u},\bm{y}+\bm{v}) \mid (\bm{x} + \bm{u}',\bm{y} + \bm{v}') \in R_{i-1}, \, \bm{x},\bm{y} \in \N^d \}$\; \label{line:saturate}
  }
  $R_i \leftarrow \Cong(R_i)$\; \label{line:ri}
  \lIf{$R_i = R_{i-1}$}{\Return{$R_i$}} \label{line:until}
}
\caption{Algorithm for bidirected reachability in PVAS.}
\label{alg:pvas}
\end{algorithm}

\Cref{alg:pvas} is a saturation algorithm that computes a semilinear representation for $R_\P$.
The sets $R_i$ are maintained by semilinear representations or Presburger formulas.
Since in this section we only prove elementary complexity bounds,
we can use both formats interchangeably.
Observe that the update in line~\ref{line:saturate} and the equality test in line~\ref{line:until} can be expressed in Presburger arithmetic.
The computation of $\Cong(\cdot)$ will be explained in the next subsection.
Consider the values of $R_i$ for $i \ge 1$ after line~\ref{line:ri} of \Cref{alg:pvas}.
They form an ascending chain of congruences $(R_i)_{i \ge 1}$, which
implies that the algorithm must terminate.
For the correctness one can prove by induction on $i$
that $(\bm{x},\bm{y}) \in R_i$ if and only if there exists a run between $(\bm{x},\varepsilon)$ and $(\bm{y},\varepsilon)$
whose stack height does not exceed $i$.
Moreover, if the algorithm terminates after $k$ iterations then $R_k = R_\P$.

We will use a primitive recursive algorithm (which is elementary in fixed
dimension) to compute $\Cong(R_i)$ from a semilinear representation for $R_i$.
Using the tools from \cite{Schmitz17} we can then prove upper
bounds for the length of the ascending chain.

\subsection{Semilinear representations for congruences}

In this section, we present an algorithm that, for a given semilinear
representation for a set $R\subseteq\N^d\times\N^d$, computes a semilinear
representation for $\Cong(R)$. Its run time is bounded by a tower of
exponentials in $\|R\|$ of height $O(d)$ (\cref{thm:cong-closure}). Note that
for bidirected VASS, it is known that in exponential space, one can compute a
semilinear representation of the reachability
set~\cite{KoppenhagenM97,DBLP:journals/entcs/BouzianeF97}. In other words, one
can compute in exponential space a representation of the congruence class of a
given vector $\bm{x}\in\N^d$. In contrast, our algorithm computes a
semilinear representation of the entire congruence.

Let the function $\exp^k$ be inductively defined by $\exp^0(x) = x$ and $\exp^{k+1}(x) = \exp^k(2^x)$.
In the following we show how to compute a semilinear representation for a congruence $\Q$
given by a semilinear basis $R \subseteq \N^d \times \N^d$ in time $\exp^{O(d)}(\|R\|)$.
In fact, we can assume that $R$ is {\em finite}
since a linear set $L = \bm{b} + P^*$ is contained in $\Cong(F_L)$
where $F_L = \{ \bm{b}, \bm{b} + \bm{p} \mid \bm{p} \in P \}$,
and, therefore, a semilinear set $\bigcup_{i=1}^m L_i$ generates the same congruence
as $\bigcup_{i=1}^m F_{L_i}$.
The semilinear representation of $\Q$ will be obtained by induction on the dimension $d$
via a decomposition of $\N^d$ into smaller regions.
A {\em region} is a linear set $L = \bm{b} + P^* \subseteq \N^d$ where $P \subseteq \{\bm{e}_1, \dots, \bm{e}_d \}$.
Its {\em dimension} is $|P|$.
In particular all sets $\bm{b} {\uparrow} = \bm{b} + \{\bm{e}_1, \dots, \bm{e}_d \}^*$ are regions.
For a region $L = \bm{b} + P^*$ and a congruence $\Q$, we define the congruence
$\Q_L = \{ (\bm{x},\bm{y}) \in (P^*)^2 \mid (\bm{b}+\bm{x},\bm{b}+\bm{y}) \in \Q \}$
on the submonoid $P^*$.

A submonoid $S \subseteq \N^k$ is {\em subtractive} if $\bm{x}, \bm{y} \in S$ and $\bm{x} \le \bm{y}$
implies $\bm{y}-\bm{x} \in S$.
For example, the non-negative restriction $G \cap \N^k$ of a group $G \subseteq \Z^k$ is a subtractive submonoid.
The following lemma is well-known, see \cite[Proposition~7.1]{EilenbergSchutzenberger1969} or \cref{app:ackermann} for a proof.

\begin{restatable}{lemma}{subtractive}
\label{lem:subtractive}
Every subtractive submonoid $S \subseteq \N^k$ is of the form $S = M^*$
where $M$ is the finite set of the minimal nonzero elements in $S$.
\end{restatable}

Eilenberg and Schützenberger observed that for every slice $S$ there exists an element $\bm{s} \in S$
such that $S - \bm{s}$ is a subtractive submonoid \cite[Proposition~7.2]{EilenbergSchutzenberger1969}.
As a consequence, for every congruence $\Q$ on $\N^d$ there exists $\bm{b} \in \N^d$
such that $\Q_{\bm{b} {\uparrow}}$ is a subtractive submonoid.
We provide an elementary bound on $\bm{b}$.

\begin{restatable}{lemma}{bigvector}
\label{lem:big-vector}
Given a finite basis $R$ for a congruence $\Q$ on $\N^d$,
one can compute in elementary time a vector $\bm{b} \in \N^d$ and a finite set $M \subseteq \N^{2d}$
such that $\Q_{\bm{b} {\uparrow}} = M^*$.
\end{restatable}

\subparagraph{Gröbner bases}
It remains to compute a semilinear representation of $\Q$ on the complement of ${\bm{b} {\uparrow}}$.
We will decompose $\N^d \setminus {\bm{b} {\uparrow}}$ into disjoint $(d-1)$-dimensional regions $L_j$,
compute bases for the restrictions $\Q_{L_j}$, and proceed inductively.
To compute the bases for $\Q_{L_j}$,
we will exploit the well-studied connection between congruences on $\N^d$ and binomial ideals \cite{MM82}.
Let $\Z[\bm{x}]$ be the polynomial ring in the variables $\bm{x} = (x_1, \dots, x_d)$ over $\Z$.
We write $\bm{x}^{\bm{u}}$ for the monomial $x_1^{\bm{u}(1)} \cdots x_d^{\bm{u}(d)}$.
An {\em ideal} is a nonempty set $I \subseteq \Z[\bm{x}]$ such that
$f,g \in I$ and $h \in \Z[\bm{x}]$ implies $f+g, hf \in I$.
An ideal $I$ is finitely represented by a {\em basis} $B$, i.e.\ $I$ is the smallest ideal containing $B$.
By Hilbert's basis theorem any ideal $I \subseteq \Z[\bm{x}]$ has a finite basis.
One of the main tools in computer algebra for handling polynomial ideals are {\em Gröbner bases},
e.g.\ to solve the ideal membership problem.
We defer the reader to \cite{BW93} for details on Gröbner bases and only mention the properties required for our purposes.
A Gröbner basis is defined with respect to an admissible monomial ordering,
e.g.\ a lexicographic ordering on the monomials.
Buchberger's algorithm \cite{Buchberger65} computes for a given basis for an ideal $I$
the {\em unique reduced} Gröbner basis for $I$ in doubly exponential space \cite{Dube90}.

A basis $R$ for a congruence $\Q$ on $\N^d$ can be translated
into the polynomial ideal $I \subseteq \Z[\bm{x}]$ generated by
$B_R = \{ \bm{x}^{\bm{u}} - \bm{x}^{\bm{v}} \mid (\bm{u},\bm{v}) \in R\}$.
It is known that $\bm{s} \sim_\Q \bm{t}$ if and only if $\bm{x}^{\bm{s}} - \bm{x}^{\bm{t}} \in I$ \cite[Lemma~1 and 2]{MM82}.
Moreover, the reduced Gröbner basis of $I$ with respect to an admissible monomial order
always consists of differences of monomials $\bm{x}^{\bm{s}} - \bm{x}^{\bm{t}}$ \cite[Theorem~2.7]{Huynh86}.

The following lemma can be reduced to two known applications of Gröbner bases.
Let $I \subseteq \Z[\bm{x}]$ be an ideal.
For a subsequence $\bm{y}$ of $\bm{x}$ we call $I \cap \Z[\bm{y}]$ the {\em elimination ideal},
which is indeed an ideal in $\Z[\bm{y}]$.
For $\bm{b} \in \N^d$ we define the {\em ideal quotient}
$I : \bm{x}^{\bm{b}} = \{ p \in \Z[\bm{x}] \mid p \bm{x}^{\bm{b}} \in I \}$,
which is also an ideal.
It is known that one can compute Gröbner bases
for $I \cap \Z[\bm{y}]$ and $I : \bm{x}^{\bm{b}}$ in elementary time \cite[Section~6]{BW93},
see~\cref{app:ackermann} for more details.

\begin{restatable}{lemma}{shift}
\label{lem:shift}
Given a finite basis $R$ for a congruence $\Q$ on $\N^d$ and a region $L \subseteq \N^d$,
one can compute in elementary time a finite basis for $\Q_L$.
\end{restatable}

We are ready to compute a semilinear representation of $\Cong(R)$ in $\exp^{O(d)}(n)$ time.
We proceed by induction over $d$.
Using \cref{lem:big-vector} we can write $\Q_{\bm{b} {\uparrow}} = M^*$.
We decompose $\N^d = \bigcup_{i=0}^m L_j$ into regions where $L_0 = \bm{b} {\uparrow}$
and $L_1, \dots, L_m$ are $(d-1)$-dimensional regions.
By \Cref{lem:shift} we can compute bases for $\Q_{L_i}$ for $i \in [1,m]$
and by induction hypothesis semilinear representations for $\Q_{L_i}$.
In this way, we obtain semilinear representations for the restrictions $Q_i = \Q \cap L_i^2$ for each $i \in [0,m]$.
Finally, we can express $\bm{s} \sim_\Q \bm{t}$ by a Presburger formula
that says that there exists a sequence of intermediate vectors of length $2(m+1)$
where adjacent elements are related by an $R$-step
or are contained in some relation $Q_i$.

\begin{restatable}{theorem}{congclosure}
\label{thm:cong-closure}
	Given a semilinear basis $R$ for a congruence $\Q$ on $\N^d$,
	one can compute a semilinear representation for $\Q$ in time $\exp^{c_1 d}(n)$ for some absolute constant $c_1$.
\end{restatable}

\subsection{Ascending chains of congruences}

\label{sec:bad-seq}

To bound the length of the chain of congruences $R_i$ in \cref{alg:pvas} we use a {\em length function theorem}
\cite[Theorem~3.15]{Schmitz17}, see also \cite{FigueiraFSS11}.
In general, such theorems allow to derive complexity bounds for algorithms
whose termination arguments are based on well-quasi orders.

\subparagraph{Fast-growing complexity classes}
In the following we state a simplified version of
\cite[Theorem~3.15]{Schmitz17}, which is sufficient for our application.  We
start by introducing fast-growing functions and complexity classes.
Recall that the Cantor normal form of an ordinal $\alpha \le \omega^\omega$
is the unique representation $\alpha = \omega^{\alpha_1}+ \dots + \omega^{\alpha_p}$
where $\alpha > \alpha_1 \ge \dots \ge \alpha_p$.
In this form $\alpha$ is a limit ordinal if and only if $p > 0$ and $\alpha_p > 0$.
A {\em fundamental sequence} for a limit ordinal $\lambda$ is a sequence $(\lambda(x))_{x < \omega}$
of ordinals with supremum $\lambda$.
Given a limit ordinal $\lambda\le\omega^\omega$ whose Cantor normal form is $\lambda = \beta + \omega^{k+1}$,
we use the standard fundamental sequence $(\lambda(x))_{x < \omega}$, defined inductively as
$\omega^\omega(x) = \omega^{x+1}$ and $(\beta + \omega^{k+1})(x) = \beta + \omega^k \cdot (x+1)$.
Given a function $h \colon \N \to \N$ the {\em Hardy hierarchy}
$(h^\alpha)_{\alpha \le \omega^\omega}$ relative to $h$ is defined by
\begin{align*}
	h^0(x) &= x, & h^{\alpha+1}(x) &= h^\alpha(h(x)), & h^\lambda(x) &= h^{\lambda(x)}(x).
\end{align*}
Using the Hardy functions $(H^\alpha)_{\alpha}$ relative to $H(x) = x+1$
we can define the {\em fast-growing} complexity classes $(\F_\alpha)_{\alpha}$ from \cite{Schmitz16}.
We denote by $\mathscr{F}_{<\alpha}$
the class of functions computed by deterministic Turing machines in time $O(H^\beta(n))$
for some $\beta < \omega^\alpha$.
By $\F_\alpha$ we denote the class of decision problems solved by deterministic Turing machines
in time $O(H^{\omega^\alpha}(p(n)))$ for some function $p \in \mathscr{F}_{<\alpha}$.
We define $\PR = \bigcup_{k < \omega} \F_k$ and $\ACK = \F_\omega$.

\subparagraph{Controlled bad sequences}
By Dickson's lemma, any sequence of vectors $\bm{x}_1, \bm{x}_2, \dots$ with $x_i \not \le x_j$ for all $i < j$
must be finite.
Such a sequence is also called {\em bad}.
To obtain complexity bounds on the length of bad sequences
we need to restrict to sequences that do not grow in an uncontrolled fashion.
In the following let $g \colon \N \to \N$ be {\em monotone}, {\em strictly inflationary},
i.e.\ $g(x) > x$ for all $x$, and {\em super-homogeneous}, i.e.\ $g(xy) \ge g(x) \cdot y$ for all $x,y \ge 1$.
A sequence of vectors $\bm{x}_0, \bm{x}_1, \dots, \bm{x}_\ell$ is {\em $(g,n)$-controlled}
if $\|\bm{x}_i\| \le g^i(n)$ for all $i \in [0,\ell]$.
The following statement follows from \cite[Theorem~3.15]{Schmitz17} and \cite[Eq.~(3.13)]{Schmitz17}.

\begin{theorem}
	\label{thm:bad-seq}
	Any $(g,n)$-controlled bad sequence over $\N^k$ has length at most $g^{\omega^k}(nk)$.
\end{theorem}

A {\em chain} in $\N^k$ is a sequence $S_0 \subsetneq S_1 \subsetneq \dots \subsetneq S_\ell$ of subsets $S_i \subseteq \N^k$.
The chain is {\em $(g,n)$-controlled} if for each $i \in [0,\ell-1]$ there exists $\bm{s}_i \in S_{i+1} \setminus S_i$
with $\|\bm{s}_i\| \le g^i(n)$.
A set $X \subseteq \N^k$ is {\em upwards closed} if $X = X {\uparrow}$.
Observe that any $(g,n)$-controlled chain of upwards closed sets in $\N^k$
is bounded by $1 + g^{\omega^k}(nk)$ since we obtain a $(g,n)$-controlled bad sequence
$\bm{s}_0, \bm{s}_1, \dots, \bm{s}_{\ell-1}$
by picking arbitrary $\bm{s}_i \in S_{i+1} \setminus S_i$ with $\|\bm{s}_i\| \le g^i(n)$.

\subparagraph{Translating congruences into upwards closed sets}
The key trick in our upper bound for ascending chains of congruences is to
translate congruences into upwards closed sets. This allows us to translate bounds 
on the length of ascending chains of upward closed sets into corresponding bounds
for congruences. The translation works as follows.
To each congruence $\Q$ on $\N^d$, we associate the upwards closed set $U(\Q) \subseteq \N^{4d}$ with
\[
	U(\Q) = \{ (\bm{x},\bm{y},\bm{u},\bm{v}) \mid (\bm{x},\bm{y}) \in \Q, \, (\bm{x}+\bm{u},\bm{y}+\bm{v}) \in \Q, \, (\bm{u},\bm{v}) \neq (\bm{0},\bm{0}) \} {\uparrow}.
\]
Clearly $\Q_1 \subseteq \Q_2$ implies $U(\Q_1) \subseteq U(\Q_2)$.
In fact, strict inclusion is also preserved:
\begin{lemma}\label{lem:translation-strictly-monotone}
	Let $\Q_1$ and $\Q_2$ be congruences with $\Q_1\subseteq\Q_2$.
	Then (i) $U(\Q_1)\subseteq U(\Q_2)$
	and (ii) for each $\bm{q}\in \Q_2\setminus \Q_1$, there is a $\bm{p}\in U(\Q_2)\setminus
	U(\Q_1)$ with $\|\bm{p}\|\le\|\bm{q}\|$.
\end{lemma}
\begin{proof}
	Statement (i) is immediate. For statement (ii) let $(\bm{s},\bm{t}) \in \Q_2 \setminus \Q_1$ be minimal.
	Since $(\bm{0},\bm{0}) \in \Q_1$ we must have $(\bm{s},\bm{t}) \neq (\bm{0},\bm{0})$
	and there exists $(\bm{x},\bm{y}) \in \Q_1$ with $(\bm{x},\bm{y})<(\bm{s},\bm{t})$.
	We choose such a vector $(\bm{x},\bm{y})$ in which $(\bm{u},\bm{v}) := (\bm{s},\bm{t}) - (\bm{x},\bm{y})$ is minimal.
	Clearly, $(\bm{x},\bm{y},\bm{u},\bm{v}) \in U(\Q_2)$.
	We claim that $(\bm{x},\bm{y},\bm{u},\bm{v}) \notin U(\Q_1)$.
	Towards a contradiction, suppose that there exists
	$(\bm{x}_1, \bm{y}_1, \bm{u}_1, \bm{v}_1) \le (\bm{x}, \bm{y}, \bm{u}, \bm{v})$
	with $(\bm{x}_1, \bm{y}_1) \in \Q_1$, $(\bm{u}_1,\bm{v}_1) \neq (\bm{0},\bm{0})$,
	and $(\bm{x}_1+\bm{u}_1,\bm{y}_1+\bm{v}_1) \in \Q_1$.
	Since $\Q_1$ is a congruence we have
	\begin{align*}
		\bm{x} + \bm{u}_1 = (\bm{x}-\bm{x}_1) + \bm{x}_1 + \bm{u}_1
		& \sim_{\Q_1} (\bm{x}-\bm{x}_1) + \bm{y}_1 + \bm{v}_1 \\
		& \sim_{\Q_1} (\bm{x}-\bm{x}_1) + \bm{x}_1 + \bm{v}_1
		= \bm{x} + \bm{v}_1 \sim_{\Q_1} \bm{y} + \bm{v}_1.
	\end{align*}
	If $(\bm{u}_1,\bm{v}_1) = (\bm{u},\bm{v})$ then this contradicts
	$(\bm{x}+\bm{u},\bm{y}+\bm{v}) = (\bm{s},\bm{t}) \notin \Q_1$.
	If $(\bm{u}_1,\bm{v}_1) < (\bm{u},\bm{v})$ then $(\bm{s},\bm{t}) =
		(\bm{x}+\bm{u}_1,\bm{y}+\bm{v}_1) + (\bm{u}-\bm{u}_1,
		\bm{v}-\bm{v}_1)$ contradicts the minimality of
		$(\bm{u},\bm{v})$.
\end{proof}

If $(\Q_i)_{i \le \ell}$ is a $(g,n)$-controlled chain of congruences in $\N^d$
then by \cref{lem:translation-strictly-monotone}, $(U(\Q_i))_{i \le \ell}$ is a $(g,n)$-controlled
chain of upwards closed subsets of $\N^{4d}$ and thus has length at
most $1+g^{\omega^{4d}}(4dn)$. Hence, the same bound applies to $(\Q_i)_{i \le \ell}$.

\begin{lemma}
	\label{lem:cong-chain}
	Any $(g,n)$-controlled chain of congruences in $\N^d$ has length $\le 1 + g^{\omega^{4d}}(4dn)$.
\end{lemma}

Putting together \Cref{thm:cong-closure} and \Cref{lem:cong-chain} we can conclude that
\Cref{alg:pvas} terminates after $H^{\omega^{4d+3}}(e(n))$ iterations for some elementary function $e(n)$.

\begin{restatable}{proposition}{bidirectedpvas}
\label{prop:bidirected-pvas}
Reachability in bidirected pushdown VAS is in $\ComplexityFont{ACKERMANN}$,
and in $\F_{4d+3}$ if the dimension $d$ is fixed.
\end{restatable}

The detailed complexity analysis can be found in \Cref{app:ackermann}.
The result above also holds for bidirected pushdown {\em VASS} (\Cref{thm:bidirected-pvass})
by simulating the state in two additional counters.
Hence the complexity for dimension $d$ increases from $\F_{4d+3}$ to $\F_{4d+11}$.

%% file: one-dim.tex
In this section we prove that reachability is in polynomial space for
bidirected 1-PVASS (\cref{thm:bidirected-pvass-1d}).  For the rest of this
section consider a 1-PVASS $\P = (Q,\Gamma,T)$ of $|Q|=n$ states.  To simplify
our bounds, we assume $|\Gamma|=2$. This can be achieved with a simple
encoding: To simulate stack letters $a_1,\ldots,a_k$, we can encode each
$a_i$ by the string $ab^iab^{k-i}a$.

\subparagraph{Preliminaries}
We extend the usual component-wise ordering $\Dominated$ to tuples  $(\Z\cup \{-\infty, \omega\})^k$.
Given two functions $f,g\colon X\to (\Z\cup \{-\infty, \omega\})^k$, we write $f\Dominated g$ to denote that $f(x) \Dominated g(x)$ for each $x\in X$.
We define the {\em one-step $\Z$ relation} $\Zstep$ for $\P$ similarly to the one-step relation $\xrightarrow{}$ but with a $\Z$-counter, i.e., we do not require the counter to remain non-negative.
A \emph{path from $p$ to $q$} consists of the initial state $p$ and a sequence of transitions of $\P$, such that it induces a run $(p_1,x_1,w_1) \Zstep (p_2,x_2,w_2) \Zstep \dots \Zstep (p_{j},x_{j},w_{j})$, with the requirement that $p_1=p$, $x_1=0$ and $w_1=w_{j}=\varepsilon$.
Given such a path $P$, we let
\begin{itemize}
\item $\MaxStackHeight(P)=\max_i |w_i|$ be the maximum stack height along $P$,
\item  $\PathSum(P)=x_j$ be the value of the counter at the end of $P$, and
\item $\PathMin(P)=\min_{i}x_i$ be the minimum value of the counter along $P$.
Note that $\PathMin(P)\leq 0$, as the counter is $0$ at the beginning of $P$.
\end{itemize}
We also write $\Reverse{P}$ for the reverse of $P$.
We denote by $\Paths{p}{q}$ the set of paths from $p$ to $q$, 
and let $\Paths{p}{q}_k=\{P\in \Paths{p}{q}\colon \MaxStackHeight(P)\leq k \}$ be the set of such paths with stack height at most $k$.
Given two paths $P_1$ and $P_2$, we write $P_1\Dominated P_2$ to denote that
$(\PathMin(P_1), \PathSum(P_1))\Dominated (\PathMin(P_2), \PathSum(P_2))$.

We say that a state $q$ is reachable from a state $p$ if  $(p,\bm{0},\varepsilon) \xrightarrow{*} (q,\bm{0},\varepsilon)$.
We say that $q$ is \emph{$\Z$-reachable} from $p$ if there is a path $P\in \Paths{p}{q}$ with $\PathSum(P)=0$ (hence state reachability implies $\Z$-reachability).
Given additionally a natural number $i \in \N$, we say that
$p$ {\em covers} $(q,i)$ if $(p,0,\varepsilon) \xrightarrow{*} (q,i+j,\varepsilon)$ for some $j\geq 0$.
Thus reachability implies coverability for $i=0$.
The following is a simpler proof of a reduction observed in~\cite{DBLP:journals/pacmpl/KP22}: For bidirected $1$-PVASS, reachability reduces to coverability and $\Z$-reachability.
\begin{lemma}\label{lem:reach_to_cov_and_zreach}
For any two states $p,q$ of $\P$, we have that $p$ reaches $q$ iff
(i)~$p$ covers $(q,0)$, 
(ii)~$q$ covers $(p,0)$, and
(iii)~$p$ $\Z$-reaches $q$.
\end{lemma}
\begin{proof}
Clearly if $p$ reaches $q$ then conditions (i)-(iii) hold, so we only need to argue about the reverse direction.
If $p$ does not cover $(q,1)$, since $p$ covers $(q,0)$, we have that $p$ reaches $q$, and we are done.
Similarly, if $q$ does not cover $(p,1)$, we have that $q$ reaches $p$ and thus we are done.
Finally, assume that $p$ covers $(q,1)$ and $q$ covers $(p,1)$, and let $P_p$ and $P_q$ be the corresponding paths witnessing coverability.
Let $P$ be a path witnessing that $p$ $\Z$-reaches $q$.
We construct the path $H_{\ell}$ witnessing the reachability of $q$ from $p$ as
$H_{\ell}=(P_p \Concat P_q)^{\ell}\Concat P \Concat (\Reverse{P_p} \Concat\Reverse{P_q})^{\ell}$,
where $\ell$ is chosen such that $\PathSum((P_p \Concat P_q)^{\ell})\geq -\PathMin(P)$.
\end{proof}
In light of \cref{lem:reach_to_cov_and_zreach}, for
\cref{thm:bidirected-pvass-1d}, it suffices to show that for bidirected
$1$-PVASS, both $\Z$-reachability and coverability can be decided in $\PSPACE$.
The former is known already: Reachability in $\Z$-PVASS belongs to
$\NP$~\cite{DBLP:conf/cav/HagueL11}; in the bidirected case, it is
even decidable in $\PTIME$~\cite{GanardiMajumdarZetzsche2022a}.  Thus, the
rest of this section is devoted to deciding coverability in $\PSPACE$.
\begin{lemma}\label{lem:1dcoverability}
Coverability in 1-dimensional bidirected pushdown VASS is in $\PSPACE$.
\end{lemma}

\subparagraph{Summary functions}
We define a summary function $\DFinite_k\colon Q\times Q\to (\Z_{\leq 0} \cup \{ -\infty \}) \times (\Z\cup\{-\infty, \omega \})$, parametric on $k\in \N$, as $\DFinite_k(p,q)=(a,b)$, where $a$ and $b$ are defined as follows.
\begin{align*}
a= \max( \{ \PathMin(P)\colon P\in \Paths{p}{q}_k \} )\qquad
b= \sup( \{ \PathSum(P)\colon P\in \Paths{p}{q}_k \text{ and } \PathMin(P)=a  \} )
\end{align*}
with the convention that $\max(\emptyset)=\sup(\emptyset) = -\infty$.
We occasionally write $\DFinite_k(p,q)=(a,\Some)$ to denote that $\DFinite_k(p,q)=(a,b)$ for some $b$.
We further define a summary function $\DOmega_k\colon V\to (\Z_{\leq 0} \cup \{ -\infty \})$, parametric on $k\in \N$, as follows.
\begin{align*}
\DOmega_k(p) = \max( \{ \PathMin(P)\colon P\in \Paths{p}{p}_k \text{ and } \PathSum(P)>0 \} )
\end{align*}
Recall that we use $n=|Q|$ for the number of states in $\P$.
It is well-known that in any pushdown system of $n$ states (and only two stack letters), a shortest path between two states has length $2^{O(n^2)}$ (e.g.\ this follows by inspecting
a proof of the pumping lemma for context-free languages~\cite[Lemma~6.1]{HU1979}; a precise bound was obtained in~\cite{Pierre1992}).
Since both the weight and the minima of any path are lower-bounded by minus the length of the path,
if $\Paths{p}{q}_{k}\neq \emptyset$, then a shortest path in  $\Paths{p}{q}_{k}$ witnesses $\DFinite_k(p,q)=(a,b)$ where $a$ and $b$ are at most exponentially negative.
This is established in the following lemma.
\begin{lemma}\label{lem:min_bounded}
Consider any two states $p,q$ and natural number $k$, and let $\DFinite_k(p,q)=(a,b)$.
If $a>-\infty$ then $a,b\geq-\beta$, for $\beta=2^{O\left(n^2\right)}$.
\end{lemma}
The bidirectedness of $\P$ also leads to the following lemma.
\begin{lemma}\label{lem:sum_bounded}
Consider any two states $p,q$ and natural number $k$, and let $\DFinite_k(p,q)=(a,b)$.
There exists a constant $\alpha$ independent of $\P$ and $k$, such that, if $b>2^{\alpha n^2}$, then $b=\omega$.
\end{lemma}

The intuition behind the summary functions $\DFinite_k$ and $\DOmega_k$ is as follows.
\cref{lem:min_bounded} and \cref{lem:sum_bounded} hint on an algorithm to decide coverability by saturating $\DFinite_k$ iteratively for increasing $k$.
The lemmas state that the finite values of $\DFinite_k$ are exponentially bounded, and thus the process is guaranteed to reach a fixpoint within exponentially many iterations.
An obstacle to this approach is that $\DFinite_k$ may fail to capture paths that are useful in subsequent iterations.
In particular, $\DFinite_k(p,q)$ misses paths that can reach $q$ with a larger counter at the cost of a lower minima on the way.
The following lemma shows that $\DOmega_k$ captures the effects of all paths missed by $\DFinite_k$, and allows the two summaries to be \emph{mutually} saturated.
\begin{lemma}\label{lem:domega_pumping}
For any states $p,q$, let $\DFinite_k(p,q)=(a,b)$, and assume there exists a path $P\in \Paths{p}{q}_k$ with $\PathSum(P)>b$.
Then $\DOmega_k(p)\geq \PathMin(P)$.
\end{lemma}
Moreover, the values of $\DOmega_k$ are also exponentially bounded, and thus the mutual saturation can be carried out in polynomial space.
\begin{lemma}\label{lem:min_cycle_bounded}
For any state $p$, if $\DOmega_k(p)>-\infty$ then $\DOmega_k(p)\geq -\beta$, for $\beta=2^{O\left(n^2\right)}$.
\end{lemma}
In the remainder of this section we describe the saturation procedure for $\DFinite_k$ and $\DOmega_k$.

\subparagraph{Finite graphs}
We consider weighted finite graphs $G=(V_G, E_G, \Weight_G)$ where $\Weight_G\colon E_G\to \Z$.
Moreover, we assume that every connected component of $G$ is strongly connected.
By a small abuse we extend some of the above notation on $\P$ to such graphs.
Given two nodes $u$, $v$ in $G$, we write $\Paths{u}{v}^G$ for the set of paths from $u$ to $v$ in $G$.
We similarly extend the summary functions to $\DFinite^{G}$ and $\DOmega^{G}$, defined by the corresponding paths $P\in \Paths{u}{v}^G$.

\begin{lemma}\label{lem:dfinite_computation}
Given a graph $G$ as above, the summary functions $\DFinite^G$ and $\DOmega^G$ can be computed in polynomial time.
\end{lemma}

\subparagraph{Computing $\DFinite_k$ and $\DOmega_k$.}
We now describe a dynamic-programming algorithm for computing $\DFinite_k$ and $\DOmega_k$, for increasing values of $k$.
We let $\StackAlphabetBot=\Gamma\cup\{\bot \}$, where $\bot$ is a special symbol, and assume without loss of generality that every transition in $\P$ that manipulates the counter does not affect the stack.
Afterwards we will argue that the algorithm terminates within exponentially many iterations.
\begin{enumerate}
\item We start with $k=0$.
We construct a graph $G_0$ that consists of nodes $\tuple{p,\bot}$  where $p$ is a state of $\P$.
Moreover, $G_0$ contains the corresponding transitions of $\P$ that do not manipulate the stack.
In particular, for every transition $(p,i,\varepsilon,q) \in T$ we have an edge $\tuple{p,\bot}\DTo{i}\tuple{q,\bot}$ in $G_0$.
We use \cref{lem:dfinite_computation} to compute $\DFinite^{G_0}$ and $\DOmega^{G_0}$, and report that, for all states $p$ and $q$, we have
$\DFinite_{0}(p,q)=\DFinite^{G_0}(\tuple{p,\bot}, \tuple{q,\bot})$ and $\DOmega_0(p)=\DOmega^{G_0}(\tuple{p,\bot})$.
\item We repeat the following until $\DFinite^{G_k}$ and $\DOmega^{G_k}$ have converged.
We construct a graph $G_k$ as follows.
For every state $p$ and every $\sigma \in \StackAlphabetBot$, we have a node $\tuple{p,\sigma}$ in $G_k$.
We then do the following.
\begin{enumerate}
\item Let $\DOmega^{G_{k-1}}(\tuple{p,\bot})=c$.
We insert a node $\tuple{p',\sigma}$, and if $-\infty < c$, we insert two edges manipulating the counter $\tuple{p,\sigma}\DTo{c} \tuple{p',\sigma}$ and $\tuple{p',\sigma}\DTo{-c+1}\tuple{p, \sigma}$.
\item For every state $q$, let $\DFinite^{G_{k-1}}(\tuple{p,\bot},\tuple{q,\bot})=(a,b)$.
If $-\infty< a$, we insert a node $\tuple{p,q,\sigma}$ and two edges manipulating the counter $\tuple{p,\sigma}\DTo{a} \tuple{p,q,\sigma}$ and $\tuple{p,q, \sigma}\DTo{-a+b'} \tuple{q,\sigma}$,
where $b'=b$ if $b<\omega$ and $b'=0$ otherwise.
\end{enumerate}
\item 
Finally, for each stack letter $\sigma\in \Gamma$, we connect nodes of the form $\tuple{p,\bot}$ and $\tuple{q, \sigma}$ using the transitions of $\P$ that manipulate the stack.
That is, for every transition $(p,0,\sigma,q) \in T$, we insert an edge $\tuple{p,\bot}\DTo{}\tuple{q,\sigma}$,
and for every transition $(p,0,\ov{\sigma},q) \in T$, we insert an edge $\tuple{p,\sigma}\DTo{}\tuple{q,\bot}$.
\item 
We use \cref{lem:dfinite_computation} to compute $\DFinite^{G_k}$ and $\DOmega^{G_k}$, and report that, for all states $p$ and $q$, we have
$\DFinite_{k}(p,q)=\DFinite^{G_k}(\tuple{p,\bot}, \tuple{q,\bot})$ and $\DOmega_k(p)=\DOmega^{G_k}(\tuple{p,\bot})$.
\end{enumerate}

We first argue that the graphs $G_k$ consist of strongly connected components, and thus \cref{lem:dfinite_computation} is applicable.

\begin{lemma}\label{lem:strongly_connected}
For all $k\in \N$, every connected component of $G_k$ is strongly connected.
\end{lemma}

Given some $\sigma\in \StackAlphabetBot$ and $k\geq 1$, we denote by $G_k\Project \sigma$ the subgraph of $G_k$ induced by the nodes whose last element is $\sigma$.
It follows directly from the construction of $G_k$ that, for every pair of states $p,q$ of $\P$ and $\sigma\in\StackAlphabetBot$, for every path $P$ that goes from $\tuple{p,\sigma}$ to $\tuple{q,\sigma}$ and is contained in $G_k\Project \sigma$, there is a path $H\in \Paths{p}{q}_{k-1}$ with $P\Dominated H$.
In turn, this implies that the summary functions $\DFinite^{G_k}$ and $\DOmega^{G_k}$ are always dominated by paths in $\P$ of stack height at most $k$, i.e., for all states $p$ and $q$, we have $\DFinite^{G_k}(\tuple{p,\bot}, \tuple{q,\bot}) \Dominated \DFinite_k(p,q)$ and $\DOmega^{G_k}(\tuple{p,\bot})\Dominated \DOmega_k(p)$ for all $k\in \N$.
The following lemma states that $\DFinite^{G_k}$ and $\DOmega^{G_k}$ compute $\DFinite_k$ and $\DOmega_k$ exactly, by arguing that $\DFinite^{G_k}$ and $\DOmega^{G_k}$ also dominate all paths of $\P$ with stack height at most $k$.
In turn, this establishes the correctness of the algorithm.

\begin{lemma}\label{lem:1d_correctness}
For all $k\in \N$ and states $p,q\in Q$, we have $\DFinite_k(p,q)=\DFinite^{G_k}(\tuple{p,\bot}, \tuple{q,\bot})$ and $\DOmega_k(p)=\DOmega^{G_k}(\tuple{q, \bot})$.
\end{lemma}

Thus, to decide whether $p$ covers $(q,0)$, we saturate $\DFinite_k$ and $\DOmega_k$ and check whether $\DFinite_k(p,q)=(0,\Some)$.
The termination and complexity of the algorithm follows from the boundedness of the finite values of $\DFinite_k$ and $\DOmega_k$ (\cref{lem:min_bounded}, \cref{lem:sum_bounded} and \cref{lem:min_cycle_bounded}), which concludes \cref{lem:1dcoverability}.

\begin{lemma}\label{lem:1d_complexity}
The above algorithm terminates and uses polynomial space.
\end{lemma}

Finally, note that if we have a polynomial bound on the stack height, then the saturation procedure runs in polynomial time, which also leads to reachability in polynomial time (a closer analysis yields an $O(n^3)$ bound per iteration).
In particular, the $\PSPACE$-hardness proof for 1-dimensional {\em directed} PVASS from \cite{EnglertHLLLS21}
cannot be directly transferred to bidirected PVASS:
The 1-PVASS constructed in \cite{EnglertHLLLS21} has a polynomially bounded stack height.
See also \cref{app:failure} for details on how exactly the reduction fails.
Without a bound on the stack height, the saturation might indeed take exponential time: There are 1-dimensional bidirected PVASS on which shortest coverability witnesses require an exponential stack height (see \cref{app:1d} for an example).

%% file: tower.tex
In this section, we show that reachability in bidirected PVASS is \ComplexityFont{TOWER}-hard with respect to elementary reductions,
and $k$-\ComplexityFont{EXPSPACE}-hard in dimension $2k + 6$.
Recall that \ComplexityFont{TOWER} is the class of all problems computable by deterministic Turing machines in time (or space) bounded by a tower of exponentials of elementary height.

\subparagraph{Lower bound for directed PVASS}
We first recall the $\TOWER$-hardness proof by Lazi\'{c} and Totzke~\cite{Lazic2017} for reachability in directed PVASS.
They reduce the $\exp^n(1)$-bounded halting problem for counter programs of size $n$,
which is $\TOWER$-complete with respect to elementary reductions (which allow us to replace the parameter $n$ with an arbitrary elementary function $e(n)$)~\cite{FischerMeyerRosenberg1968}.
A {\em counter program} is a finite sequence of commands which manipulate non-negative counters, initially set to zero.
The commands include increments $x := x + 1$, decrements $x := x - 1$, conditionals {\bf if} $x = 0$ {\bf then goto} $L_1$ {\bf else goto} $L_2$ (where $L_1$ and $L_2$ are line numbers),
and ${\bf halt}$.
If a counter of value 0 is decremented, the program aborts.
The $\exp^n(1)$-bounded halting problem asks whether given a counter program of size $n$,
starting from the first command and all counters set to zero,
a command ${\bf halt}$ can be reached using a run during which all counters
are bounded by $\exp^n(1)$ and are all zero at the end.

As in most lower bounds for vector addition systems and their extensions,
for each counter $x$ we store a complement counter $\bar x$,
maintaining the invariant $x + \bar x = B$ for some (large) bound $B$.
This can be achieved by complementing every in/decrement of $x$ by a de/increment of $\bar x$, and vice versa.
Then, a zero test {\bf if} $x = 0$ {\bf then goto} $L_1$ {\bf else goto} $L_2$ can be replaced by guessing
whether $x = 0$ and $x \neq 0$:
In the former case we add and subtract $B$ to $x$ and continue with $L_1$.
In the latter case we add and subtract $B$ to $\bar x$ and continue with $L_2$.

The challenge is to implement the addition (and subtraction) with a large number $B$, here $B = \exp^n(1)$, using a polynomially large system.
Suppose we have counters $c_1, \dots, c_n$ with complement counters $\bar c_1, \dots, \bar c_n$
satisfying $c_k + \bar c_k = \exp^k(1)$ for all $k$.
Lazi\'{c} and Totzke~\cite{Lazic2017} show how to construct for all $k = 1, \dots, n$
a $\text{poly}(k)$-sized PVASS that adds $\exp^k(1)$ to $c_k$.
It operates by pushing exactly $\exp^{k-1}(1)$ many zeros to the stack,
repeatedly incrementing the $\exp^{k-1}(1)$-bit binary counter on the stack, while simultaneously decrementing $c_k$,
and finally popping exactly $\exp^{k-1}(1)$ many ones from the stack.
Observe that the operations on the stack can be implemented with the help of $c_{k-1}$ that can be in/decremented by $\exp^{k-1}(1)$
by induction hypothesis.
Before simulating the counter program, each complement counter $\bar c_k$ has to be set to $\exp^k(1)$,
which can be done in a similar fashion.

\subparagraph{Simulation by bidirected systems}
In the following we will make the above construction outlined by Lazi\'{c} and Totzke~\cite{Lazic2017} explicit
and show that the simulation is correct even after making the PVASS bidirected.
To this end we need the following argument already found in Post's undecidability proof
of the word problem for Thue systems \cite[Lemma~II]{Post47}.
Consider a deterministic transition system where the final state does not have any outgoing transitions.
To such a system we now add reverse edges to make it bidirected.
Clearly, any original run is present in the bidirected system.
Conversely, consider a run to the final state in the bidirected system, which may use the new reverse edges.
It cannot end on a reverse edge, since there is no outgoing forward edge from the final state. So as long as the run contains reverse edges, at least one of these edges must be followed by one in the forward direction. Let us call them $p \xrightarrow{\bar a} q$ and $q \xrightarrow{b} r$. As the original system was deterministic $q$ has exactly one outgoing edge, and hence $(q \xrightarrow{b} r) = (q \xrightarrow{a} p)$. Since the effects of $\bar a$ and $b$ cancel out, we can omit both of them from the run. Iterating this argument eventually yields a run with no reverse edges. It follows that adding reverse edges to a fully deterministic system does not change its reachability set (this was originally shown by Mayr and Meyer \cite{MM82} for their proof of \ComplexityFont{EXPSPACE}-hardness of reachability for bidirected VAS).

The construction of Lazi\'{c} and Totzke is not fully deterministic.
However, it only uses very restricted nondeterminism that will not impact our simulation of the counter machine.

\subparagraph{Gadget construction}

Since we also want to show a $k$-$\EXPSPACE$ lower bound for dimension $2k+6$, we use a slightly more refined analysis: We will assume that two numbers $k$ and $n$ are given as input and construct a system that simulates counters bounded by $\exp^k(n)$ instead of $\exp^k(1)$ as in Lazi\'{c} and Totzke.

In the following a {\em gadget} $G$ consists of a PVASS and two distinguished terminal states $s$ and $t$.
We consider vectors $\bm{x} \in \N^{2k}$ where the first $k$ components are viewed the values of $k$ counters $c_1, \dots, c_k$
and the last $k$ components are the values of $k$ complementary counters $\bar c_1, \ldots, \bar c_k$.
Without further mention, any update on a counter $c$ is always understood
with complementary update on $\bar c$ so that the sums $c_i + \bar c_i$ remain constant.

Given two numbers $k,n$ (in unary), we will inductively construct a gadget $G_k$ with stack alphabet $\Gamma_k$.
This gadget will allow us to add $\exp^k(n)$ to a counter.
The gadget's size will grow exponentially in $k$ (and polynomially in $n$), and later, we improve the construction to grow only polynomially in $k$.
The gadget $G_1$ simply increments $c_1$ by $2^n$.
Assuming $G_{k-1}$ is already constructed, we construct the gadget $G_k$.
The gadget $\bar G_{k-1}$ is obtained from $G_{k-1}$ by reversing all transitions, and interchanging its terminal states.
Its behavior is inverse to that of $G_{k-1}$, as it subtracts $\exp^{k-1}(n)$ from $c_{k-1}$.
Let $Z_{k-1} = G_{k-1} \circ \bar G_{k-1}$ be the gadget obtained by composing $G_{k-1}$ with $\bar G_{k-1}$,
which is a zero test of $c_{k-1}$.
We can naturally view $G_{k-1}$, $\bar G_{k-1}$ and $Z_{k-1}$ as gadgets with $2k$ counters,
where $c_k$ and $\bar c_k$ are untouched.
The gadget $G_k$ is displayed in \Cref{fig:gk} where $\texttt{0}$ and $\texttt{1}$ are fresh stack symbols
and Inc$_k$ is a subprocedure which increments the binary counter on the stack.

To prove correctness of the gadget $G_k$ we need a bit of notation.
For brevity we write $[x_1, \dots, x_k]$ for $(x_1, \dots, x_k, \exp^1(n) - x_1, \dots, \exp^k(n) - x_k)$.
Our gadgets will always assume that the ``lower'' counters $c_j$ are set to zero and that the invariant is satisfied.
A counter vector of the form $[0, \dots, 0, x_i, \dots, x_k]$ is called {\em $i$-initialized}.
Moreover, we call a run $(s,\bm{u},w) \xrightarrow{*} (t,\bm{v},w')$ in a gadget \emph{$i$-initialized}
if either $\bm{u}$ or $\bm{v}$ is $i$-initialized.

\begin{figure}
	\begin{center}
	\begin{tikzpicture}[state/.style={fill, circle, inner sep=2pt}, ->, >=stealth]
			\node (aa) {Inc$_k$:};
			\node[state, label={below:$p_0$}, right = 1cm of aa] (a) {};
			\node[state, label={below:$p_1$}, right = 2cm of a] (b) {};
			\node[state, label={below:$p_2$}, right = 2cm of b] (c) {};
			\node[state, label={below:$p_3$}, right = 2cm of c] (d) {};
			\node[state, label={below:$p_4$}, right = 2cm of d] (e) {};
			\draw (a) edge[above] node{$Z_{k-1}$} (b);
			\draw[loop above] (b) edge[above] node[align=center]{\footnotesize $c_{k-1} +$= 1 \\ \footnotesize $\bar{\texttt{1}}$} (b);
			\draw (b) edge[above] node[align=center]{\footnotesize $\bar{\texttt{0}} \texttt{1}$} (c);
			\draw[loop above] (c) edge[above] node[align=center]{\footnotesize $c_{k-1} -$= 1 \\ \footnotesize $\texttt{0}$} (c);
			\draw (c) edge[above] node{$Z_{k-1}$} (d);
			\draw (d) edge[above] node{\footnotesize $c_k +$= 1} (e);
		\end{tikzpicture}
		
		\vspace{1em}
		
		\begin{tikzpicture}[state/.style={fill, circle, inner sep=2pt}, ->, >=stealth]
			\node[state, label={below:$q_0$}] (a) {};
			\node[state, label={below:$q_1$}, right = 2cm of a] (b) {};
			\node[state, label={below:$q_2$}, right = 2cm of b] (c) {};
			\node[state, label={below:$q_3$}, right = 2cm of c] (d) {};
			\node[state, label={below:$q_4$}, right = 2cm of d] (e) {};
			\node[state, label={below:$q_5$}, right = 2cm of e] (f) {};
			\draw (a) edge[above] node{$G_{k-1}$} (b);
			\draw[loop above] (b) edge[above] node[align=center]{\footnotesize $c_{k-1} -$= 1 \\ \footnotesize $\texttt{0}$} (b);
			\draw (b) edge[above] node{$Z_{k-1}$} (c);
			\draw[loop above] (c) edge[above] node[align=center]{$\text{Inc}_k$} (c);
			\draw (c) edge[above] node{$Z_{k-1}$} (d);
			\draw[loop above] (d) edge[above] node[align=center]{\footnotesize $c_{k-1} +$= 1 \\ \footnotesize $\bar{\texttt{1}}$} (d);
			\draw (d) edge[above] node{\footnotesize $c_k +$= 1} (e);
			\draw (e) edge[above] node{$\bar G_{k-1}$} (f);
		\end{tikzpicture}
	\end{center}
	\caption{Gadgets Inc$_k$ and $G_k$.}
	\label{fig:gk}
\end{figure}

\begin{restatable}{proposition}{directedgadget}
The $k$-initialized runs in $G_k$ from $q_0$ to $q_5$
are precisely the runs
\[
	(q_0,[0, \dots, 0],w) \xrightarrow{*}_{G_k} (q_5,[0, \dots, 0, \exp^k(n)],w) \quad \text{for } w \in \Gamma_{k-1}^*.
\]
\end{restatable}

Next we will analyse the bidirected version of $G_k$.
In order to distinguish the original transitions from the reverse transitions
we define for a PVASS $G$ the relations $\leftrightarrow_G \,=\, \to_G \cup \leftarrow_G$
and $\stackrel{*}{\leftrightarrow}_G$, denoting the reflexive transitive closure of $\leftrightarrow_G$.
Similarly to the argument by Post~\cite[Lemma~II]{Post47}, we can prove the following:

\begin{restatable}{proposition}{bidirectedgadget}
	\label{prop:bidirectedgadget}
	Let $\bm{u}, \bm{v} \in \N^{2k}$ where $\bm{u}$ or $\bm{v}$ is $k$-initialized.
	\begin{itemize}
	\item If $(q_0,\bm{u},w)  \stackrel{*}{\leftrightarrow}_{G_k} (q_5,\bm{v},w')$
	then $(q_0,\bm{u},w)  \xrightarrow{*}_{G_k} (q_5,\bm{v},w')$.
	\item If $q \in \{q_0,q_5\}$ and $(q,\bm{u},w) \stackrel{*}{\leftrightarrow}_{G_k} (q,\bm{v},w')$
	then $\bm{u} = \bm{v}$ and $w = w'$.
	\end{itemize}
\end{restatable}

We need to reduce the size of $G_k$ so that it can be constructed in time polynomial in~$k$.
Since $G_k$ uses ten copies of the subgadget $G_{k-1}$ (each zero test $Z_{k-1}$ uses two copies of $G_{k-1}$),
we cannot simply insert $G_{k-1}$ by copying it, as this would induce exponential growth of the number of states of our system.
Instead, we instantiate each gadget $G_{k-1}$ once.
Then, whenever a gadget would be used between two states $p, q$,
we push a fresh stack symbol $t_{p,q}$ and move to $G_{k-1}$.
When exiting $G_{k-1}$ we pop $t_{p,q}$ and return to $q$.
Since this symbol is unique for every pair of states,
it uniquely determines where we can leave the gadget to, even if there are multiple incoming and outgoing transitions
at the gadget $G_{k-1}$.
Finally, one can verify that \cref{prop:bidirectedgadget} still holds for this adapted version of $G_k$.

\subparagraph{Simulating the counter program}

We are ready to finish the lower bound proof.
We are given a counter program of size $n$ with three counters $x_1,x_2,x_3$ and want to reduce the $\exp^{k+1}(n)$-bounded halting problem
to the reachability problem for bidirected PVASS using $2k+6$ counters.
To this end, we construct the gadget $G_{k+1}$ three times: Each of these three instances has, instead of $c_{k+1}$ (and its complement), a counter $x_i$ (and its complement) for some $i\in\{1,2,3\}$. 
However, the three instances of $G_{k+1}$ share the counters $c_1,\ldots,c_k$ (and their complements).
Thus, in total, we have $2\cdot k+2\cdot 3=2k+6$ counters.
If $k$ is fixed, this yields our $k$-$\EXPSPACE$ lower bound (\cite{FischerMeyerRosenberg1968}). If $k$ is part of the input, the problem becomes $\TOWER$-complete.
We start by initializing the complement counters $\bar c_1, \dots, \bar c_k$ in sequence,
using variants of the gadgets $G_i$ that 
(i) operate on the balance counter $\bar c_i$ instead of $c_i$,
(ii) do not decrement $c_i$ when incrementing $\bar c_i$,
and (iii) operate on the lower $i-1$ counters as normal.
Similarly we initialize $\bar x_1, \bar x_2, \bar x_3$ to $\exp^{k+1}(n)$.
Finally, in order to have an all-zero configuration in the final state, we de-initialize these counters before entering the final state.

Increments and decrements in the counter program are directly translated into counter updates in the PVASS.
A conditional {\bf if} $x_i = 0$ {\bf then goto} $L_1$ {\bf else goto} $L_2$ is replaced by a nondeterministic guess
of whether $x_i = 0$ or $x_i \neq 0$, verifying this (in)equality, and jumping to $L_1$ or $L_2$.
Here we use variants of the zero tests $Z_{k+1} = G_{k+1} \circ \bar G_{k+1}$
which on their highest level operate on $x$ and $\bar x$ (instead of $c_{k+1}$ and $\bar c_{k+1}$).
The question of reachability of ${\bf halt}$ is then a reachability instance on the bidirected version of the PVASS.

If the counter program halts then we can find a corresponding computation in the PVASS.
Conversely, consider a successful run of the bidirected PVASS which uses a minimal number of reverse transitions.
By \Cref{prop:bidirectedgadget} we can assume
that no gadget $G_{k+1}$ and $\bar G_{k+1}$ (and their variants) is entered and exited through the same terminal state.
Furthermore, any subrun passing through such a gadget can be assumed to use only forward transitions.
Hence the only reverse transitions remaining are from increments or decrements.
Observe that the last occurrence of such a reverse transition $\bar \tau$
must be followed by its corresponding forward transition $\tau$.
Hence we can cancel $\tau$ with $\bar \tau$, contradiction.

%% file: conclusion.tex
We have shown that the reachability problem in bidirected
pushdown VASS is decidable, with an Ackermann upper bound
and a $\TOWER$ lower bound. Moreover, in the
one-dimensional case, the problem is in $\PSPACE$, whereas
$\PTIME$-hardness was shown in
\cite{GanardiMajumdarZetzsche2022a}. Thus, the exact
complexity, both in the general and the one-dimensional
case, remains open.

Another direction for future research is to study
bidirected versions of other infinite-state models. For
example, pushdown VASS are the simplest level in a
hierarchy of infinite-state models for which decidability
of the reachability problem is open~\cite{Zetzsche2017a}.
Perhaps the techniques from this paper can be applied to
show decidability of all levels in the bidirected setting.

%% file: app-ackermann.tex
\subtractive*

\begin{proof}

Let $M \subseteq S$ be the set of minimal nonzero elements of $S \subseteq \N^k$.
Clearly $M^* \subseteq S$ because $S$ is a submonoid.
For the converse, let $\bm{s}$ be a minimal element in $S \setminus M^*$.
Then, there exists $\bm{m} \in M^*$ such that $\bm{m} < \bm{s}$.
Since $S$ is subtractive we know that $\bm{s} - \bm{m} \in S$.
Moreover, by minimality of $\bm{s}$ in $S \setminus M^*$
we obtain $\bm{s} - \bm{m} \in M^*$.
But then $\bm{s} = (\bm{s} - \bm{m}) + \bm{m} \in M^*$ is a contradiction.
\end{proof}

\bigvector*

\begin{proof}

	We may clearly assume that $R$ is symmetric.
Let $\S = \langle V \rangle \cap (\N^d \times \N^d)$
where $V = R \cup E$,
$E = \{ (\bm{e}_i,\bm{e}_i) \mid i \in [1,d] \}$,
and $\langle \cdot \rangle$ denotes the generated subgroup.
We claim that $\S$ is a congruence on $\N^d$.
First, $\S$ is a subtractive submonoid and therefore $(\bm{x},\bm{y}), (\bm{z},\bm{z}) \in \S$ implies
$(\bm{x}+\bm{z},\bm{y}+\bm{z}) \in \S$.
It is a reflexive and symmetric relation since $E^* \subseteq \S$ and $V$ is symmetric.
For transitivity suppose $(\bm{x},\bm{y}), (\bm{y},\bm{z}) \in \S$,
and therefore $(\bm{x}+\bm{y},\bm{y}+\bm{z}) \in \S$.
Since $(\bm{y},\bm{y}) \in \S$ and $\S$ is subtractive we also have $(\bm{x},\bm{z}) \in \S$.

Next observe that $\Q \subseteq \S$ since $\S$ is a congruence containing $R$.
Furthermore $\Q_{\bm{b} {\uparrow}} \subseteq \S$ for any $\bm{b}$
since the subgroup $\langle V \rangle$ is clearly invariant under translation by $(\bm{b},\bm{b})$.
It remains to find a vector $\bm{b}$ such that $\Q_{\bm{b} {\uparrow}} \supseteq \S$.
Let $A \in \Z^{2d \times 2|V|}$ be the matrix with the $2|V|$ many column vectors $\{\bm{v}, -\bm{v} \mid \bm{v} \in V\}$.
Then we can write:
\[
	\S = \{ \bm{y} \in \N^{2d} \mid \exists \bm{x} \in \N^{2|V|} \colon A\bm{x} = \bm{y} \}.
\]
The set of all solutions $(\bm{x}, \bm{y}) \in \N^{2d + 2|V|}$ of $A \bm{x} - \bm{y} = \bm{0}$ forms
a submonoid generated by the finite set $M'$ of its minimal nonzero solutions.
By \cite[Theorem~1]{Pottier91} the norm of each minimal solution in $M'$ is bounded by
$\lambda = (1 + \|A\|_{1,\infty})^{2d}$ where $\|A\|_{1,\infty}$ is the maximal 1-norm of a column in $A$.
By projecting $M'$ to the $\bm{y}$-components we obtain the set $M$ of minimal nonzero solutions of $\S$,
which generates $\S$.
Hence, each vector $(\bm{s},\bm{t}) \in M$ can be written as
\begin{equation}
	\label{eq:st}
	(\bm{s},\bm{t}) = \sum_{\bm{v} \in V} \lambda_{\bm{s},\bm{t},\bm{v}} \cdot \bm{v} \quad \text{for some } \lambda_{\bm{s},\bm{t},\bm{p}} \in \Z, \, \bm{v} \in V
\end{equation}
where the absolute values of the coefficients $\lambda_{\bm{s},\bm{t},\bm{p}}$ are bounded by $\lambda$.
Define the vector
\begin{equation}
	\label{eq:r}
	(\bm{b},\bm{b}) = \lambda \cdot \sum_{\bm{v} \in V} \bm{v}.
\end{equation}
By summing \eqref{eq:st} and \eqref{eq:r}, we can write $(\bm{b},\bm{b})+ (\bm{s},\bm{t})$
as a non-negative linear combination of vectors $\bm{p} \in V$.
This proves $(\bm{s},\bm{t}) \in \Q_{\bm{b} {\uparrow}}$ and hence also $\S = M^* \subseteq \Q_{\bm{b} {\uparrow}}$.
\end{proof}

\subparagraph{Elimination ideals and ideal quotients}

Let $I \subseteq \Z[\bm{x}]$ is be an ideal, given by a basis $B$, and $\bm{y}$ is a subsequence of $\bm{x}$.
To compute a Gröbner basis for the elimination ideal $I \cap \Z[\bm{y}]$,
we first turn $B$ into a Gröbner basis with respect to a monomial ordering
where the variables in $\bm{y}$ are ordered greater than the variables in $\bm{x} \setminus \bm{y}$.
Then $B \cap \Z[\bm{y}]$ is a Gröbner basis for $I \cap \Z[\bm{y}]$ \cite[Proposition~6.15]{BW93}.

Next we explain how to compute a Gröbner basis for the ideal quotient $I : \bm{x}^{\bm{b}}$
for a given vector $\bm{b} \in \N^d$.
Observe that $I : \bm{x}^{\bm{b}}$ is a projection of the set of all solutions $(p_0,\bm{q}) \in \Z[\bm{x}]^{1+|B|}$
of the polynomial equation $p \bm{x}^{\bm{b}} + \sum_{f \in B} q_f f = 0$.
In general, if $\bm{f}=(f_1,\ldots,f_n)\in\Z[\bm{x}]^n$ is a tuple of polynomials, then the solutions $(q_1, \dots, q_n) \in \Z[\bm{x}]^n$, called {\em syzygies},
to the polynomial equation $q_1 f_1 + \dots + q_n f_n = 0$
form a {\em submodule} of $\Z[\bm{x}]^n$, denoted $\mathrm{Syz}(\bm{f})$. This means, $\mathrm{Syz}(\bm{f})$ is closed under componentwise addition in $\Z[\bm{x}]^n$
and is closed under scalar multiplication by $\Z[\bm{x}]$-polynomials.
Given $\bm{f} \in \Z[\bm{x}]^n$, one can compute a {\em generating set} for $\mathrm{Syz}(\bm{f})$
in elementary time \cite[Section~6.1]{BW93}, i.e. tuples $\bm{g}_1, \dots, \bm{g}_m \in \Z[\bm{x}]^n$
such that $\mathrm{Syz}(\bm{f}) = \{ \sum_{i=1}^m h_i \bm{g}_i \mid h_1, \dots, h_m \in \Z[\bm{x}] \}$.

\shift*

\begin{proof}
Let $L = \bm{b} + P^*$ be the region. We can treat the cases $L = \bm{b} {\uparrow}$ and $L = P^*$ separately
since $\Q_L = (\Q_{\bm{b} {\uparrow}})_{P^*}$.
For $L = \bm{b} {\uparrow}$ we compute the reduced Gröbner basis $B_{\bm{b}}$ for the ideal quotient $I : \bm{x}^{\bm{b}}$
and turn the basis back into a finite congruence basis
$R_{\bm{b}} = \{ (\bm{s},\bm{t}) \mid \bm{x}^{\bm{s}} - \bm{x}^{\bm{t}} \in B_{\bm{b}} \}$.
For $L = P^*$ we compute the reduced Gröbner basis $B_{\bm{y}}$ for the elimination ideal $I \cap \Z[\bm{y}]$
where $\bm{y}$ contains the variables in $\bm{x}$ that do not correspond to a unit vector in $P$.
Then we turn the basis $B_{\bm{y}}$ back into a finite congruence basis
$R_{\bm{y}} = \{ (\bm{s},\bm{t}) \mid \bm{x}^{\bm{s}} - \bm{x}^{\bm{t}} \in B_{\bm{y}} \}$.
\end{proof}

\congclosure*

\begin{proof}
As described before we can assume that $R$ is finite by replacing each linear set
$\bm{b} + \{\bm{p}_1, \dots, \bm{p}_n\}^*$
by the vectors $\bm{b}, \bm{b} + \bm{p}_1, \dots, \bm{b} + \bm{p}_n$.
We proceed by induction over $d$.
It suffices to give a reduction from an instance in dimension $d$ to instances in dimension $d-1$
whose running time is bounded by an elementary function independent from $d$.

First, we compute a vector $\bm{b} \in \N^d$ and a finite set $M \subseteq \N^{2d}$ of vectors of elementary norm
such that $\Q_{\bm{b} {\uparrow}} = M^*$ by \cref{lem:big-vector}.
This yields a representation for $\Q \cap L_0^2 = (\bm{b},\bm{b}) + M^*$ on the region $L_0 = \bm{b} {\uparrow}$.
The complement $\N^d \setminus (\bm{b} {\uparrow})$ can be decomposed into at most
$d \cdot \|\bm{b}\|$ many regions of dimension $d-1$:
Let $i_1, i_2, \dots, i_n \in [1,d]$ be arbitrary such that $\bm{b} = \bm{e}_{i_1} + \dots + \bm{e}_{i_n}$
where $n \le d \cdot \|\bm{b}\|$.
Then $\N^d \setminus (\bm{b} {\uparrow}) = \bigcup_{j=1}^n L_j$
where the region $L_j = \bm{b}_j + P_j^*$ is defined by
$\bm{b}_j = \bm{e}_{i_1} + \dots + \bm{e}_{i_{j-1}}$
and $P_j = \{ \bm{e}_k \mid k \neq i_j \}^*$.
Using \cref{lem:shift} we compute bases for the congruences $\Q_{L_j}$ on the submonoids $P_j^*$.
By projecting away some zero-component we can compute a semilinear representation for $\Q_{L_j}$ by induction hypothesis,
and thus also for $\Q \cap L_j^2 = (\bm{b}_j,\bm{b}_j) + \Q_{L_j}$.

Finally, given a decomposition $\N^d = \bigcup_{j=0}^n L_j$ into regions
and semilinear representations for the restrictions $Q_j = \Q \cap L_j^2$,
we can define the entire congruence $\Q$ in Presburger arithmetic.
Suppose that $\bm{s} \sim_\Q \bm{t}$
and consider a minimal length sequence 
$\bm{w}_1, \dots, \bm{w}_\ell \in \N^d$ where $\bm{s} = \bm{w}_1$, $\bm{w}_\ell = \bm{t}$,
and for each $i \in [1,\ell-1]$,
either there exists $(\bm{u},\bm{v}) \in R$ with
$\bm{w}_i - \bm{u} \ge \bm{0}$ and $\bm{w}_{i+1} = \bm{w}_i - \bm{u} + \bm{v}$,
or $(\bm{w}_i,\bm{w}_{i+1}) \in Q_k$ for some $k \in [0,n]$.
Then for each region $L_k$ the derivation contains at most 2 elements $\bm{w}_i$ belonging to $L_k$
since we could otherwise replace a subsequence $\bm{w}_i, \bm{w}_{i+1}, \dots, \bm{w}_j$ where $\bm{w}_i, \bm{w}_j \in L_k$
by the pair $\bm{w}_i, \bm{w}_j$ related in $Q_k$.
Therefore, we can express $\bm{u} \sim_\Q \bm{v}$ in Presburger arithmetic
by quantifying over sequences of length at most $2(n+1)$
where adjacent elements are related by an $R$-step or in some relation $Q_k$.
This concludes the proof.
\end{proof}

\bidirectedpvas*

\begin{proof}
	Let us bound the length of the chain $R_1 \subsetneq R_2 \subsetneq \dots \subsetneq R_\ell$ computed by \Cref{alg:pvas}.
	Let $n$ be the size of the PVAS.
	Then $\|R_i\| \le \exp^{c_2 d}(\|R_{i-1}\|)$ for all $i \in [1,\ell]$ where $c_2$ is an absolute constant.
	This follows from \cref{thm:cong-closure} and the fact that the update in line 7 takes elementary time.
	Hence, we can bound $\|R_i\| \le \exp^{c_2 i d}(n)$ for all $i \in [1,\ell]$.
	This in turn implies that the chain $R_1, R_2, \dots$ is $(g,n)$-controlled
	where $g(n) = \exp^{c_3 d}(n)$ for some absolute constant $c_3$.
	Here we use the fact that the set difference $S_2 \setminus S_1$
	of two semilinear sets contains a vector of elementary norm in $\|S_1\|+\|S_2\|$,
	if it is empty.
	It is easy to verify that $\exp(x) = 2^x \le 2^{x+1}(x+1)-1 = H^{\omega^2}(x)$.
	Since the Hardy functions satisfy $(h^\alpha)^\beta(x) = h^{\alpha \cdot \beta}(x)$
	(under specific conditions on $\alpha \cdot \beta$, see \cite[Eq.~(3.12)]{Schmitz16})
	the chain is indeed $(H^{\omega^2 c_3 d},n)$-controlled.
	By \Cref{lem:cong-chain} its length $\ell$ is bounded by
	\begin{align*}
	1 + (H^{\omega^2 c_3 d})^{\omega^{4d}}(4dn) &\le 1 + (H^{\omega^2 c_3 d})^{\omega^{4d}}(4dn c_3) \\
	&\le 1 + (H^{\omega^3})^{\omega^{4d}}(4dnc_3) \\
	& = 1 + H^{\omega^{4d+3}}(4dnc_3).
	\end{align*}
	
	It remains to argue that the running time is basically determined by $\ell$.
	Iteration $i$ in \Cref{alg:pvas} takes $\exp^{O(d)}(\|R_{i-1}\|)$ time,
	where $\|R_{i-1}\|$ can be bounded by $\exp^{O(\ell d)}(n)$.
	Hence the total running time of the algorithm is $\ell \cdot \exp^{O(\ell d)}(n)$.
	By \cite[Lemma~4.6]{Schmitz16} we obtain the complexity bounds $\F_{4d+3}$ in fixed dimension $d$
	and $\F_\omega$ in arbitrary dimension.
\end{proof}

%% file: app-1d.tex
\begin{lemma}\label{lem:summary_functions}
Consider any two states $p,q$ and natural number $k\in \N$,
and let $\DFinite_k(p,q)=(a,b)$ and $\DOmega_k(p)=c$, with $-\infty< a$.
Either $a>c$, or $a=c$ and $b=\omega$.
\end{lemma}
\begin{proof}
Since $-\infty< a$, there is a path $P\in \Paths{p}{q}_k$.
Assume that $a \leq c$, and let $C\in \Paths{p}{p}_k$ be a path with $\Weight(C)>0$.
Then for any natural number $\ell\in \N$, we can construct a path $H_{\ell}$ with $\PathMin(H_{\ell})=c$ and $\PathSum(H_{\ell})\geq\ell$ by traversing $C$ sufficiently many times and then following $P$ to $q$.
The sequence of paths $(H_{\ell})_{\ell }$ witnesses that $a=c$ and $b=\omega$, as desired.
\end{proof}

\begin{proof}[Proof of \cref{lem:sum_bounded}]
Since $b>-\infty$, we have $\Paths{p}{q}_k\neq \emptyset$.
Since $\P$ is bidirected, we also have $\Paths{q}{p}_k\neq \emptyset$, and thus there exists a path $P'\in \Paths{q}{p}_k$.
It is known that shortest paths in pushdown systems of $n$ states have length upper bounded by $2^{O(n^2)}$ (this follows from the pumping lemma for context-free languages~\cite[Lemma~6.1]{HU1979}), hence $|P'|\le2^{\alpha n^2}$, for some constant $\alpha$.
Now assume that $b>|P'|$, and thus we have a path $P\in \Paths{p}{q}_k$ with $\PathMin(P)=a$ and $\PathSum(P)>|P'|$.
We construct a cycle $C\in \Paths{p}{p}_k$ as $C=P\Concat P'$.
We have $\PathSum(C)=\PathSum(P)+\PathSum(P')\geq\PathSum(P) - |P'|>0$ and $\PathMin(C)\geq \min(\PathMin(P), \PathSum(P)-|P'|)=a$.
Thus, there exists a cycle starting in $p$ whose minimal counter value is $\ge a$ and which has
positive effect.
Therefore, if we set $c=\DOmega_k(p)$, then we have shown $a\leq c$. By \cref{lem:summary_functions} we have $a=c$ and $b=\omega$.
\end{proof}

\begin{proof}[Proof of \cref{lem:domega_pumping}]
By the definition of $\DFinite_{k}$, we have that $\PathMin(P)<a$.
Let $H$ be a path that witnesses $\DFinite_{k}(p,q)$, and $C=P\Concat \Reverse{H}$, 
and notice that $C\in \Paths{p}{p}_{k}$.
Moreover, $\PathSum(C) = \PathSum(P) + \PathSum(\ov{H}) = \PathSum(P)-b>0$,
while $\PathMin(C)=\min(\PathMin(P), \PathSum(P) + \PathMin(\Reverse{H}))$.
Note that $\PathMin(\Reverse{H}) = a-b$, and thus $\PathSum(P)+\PathMin(\Reverse{H})>a$, concluding that $\PathMin(C)=\PathMin(P)$.
Thus $\DOmega_k(p)\geq \PathMin(P)$, as desired.
\end{proof}

\begin{proof}[Proof of \cref{lem:min_cycle_bounded}]
Consider a shortest path $P\in \Paths{p}{p}_k$ with $\PathSum(P)>0$, and we first argue that $\MaxStackHeight(P)\leq 2n^2$.
Indeed, if $\MaxStackHeight(P)> 2n^2$, then there is a pair of states $(q, r)$ that appears three times in $P$, $(q_i,r_i)_{1\leq i \leq 3}$, such that, for each $i\in \{1,2\}$, $P$ has subpaths $P^q_i\colon q_{i}\rightsquigarrow q_{i+1}$ and $P^r_i\colon r_{i+1}\rightsquigarrow r_i$ such that removing $P^q_i$ and $P^r_i$ yields another path $P'_i\in \Paths{p}{p}_{<k}$.
Moreover, let $P'_{1,2}\in\Paths{p}{p}_{<k}$ be the path obtained by removing $P^q_i$ and $P^r_i$ for both $i=1$ and $i=2$.
Then each of the three paths $P'_1,P'_2,P'_{1,2}$ is strictly shorter than $P$ and therefore $w(P'_1)=w(P'_2)=w(P'_{1,2})=0$: If one of these paths had non-zero effect, then this path or its reverse would be a shorter cycle on $p$ than $P$ with positive effect.
However, observe that 
\[ w(P)=w(P'_{1,2})~+~\underbrace{(w(P'_1)-w(P'_{1,2}))}_{\text{effect of $P^q_1$ and $P^r_1$}}~+~\underbrace{(w(P'_2)-w(P'_{1,2}))}_{\text{effect of $P^q_2$ and $P^r_2$}}, \]
and since the right-hand side vanishes, we have $w(P)=0$, a contradiction.
Thus $\MaxStackHeight(P)\leq 2n^2$.

Next, we argue that $|P| \le 2n|\Gamma|^{2n^2}$. Indeed, if this is not the case, then $P$ traverses two cycles in the configuration space, i.e., it contains two subpaths $H_s\in \Paths{s}{s}$ and $H_t\in \Paths{t}{t}$ such that, removing either or both of them, results in another path in $\Paths{p}{p}_{k'}$, where $k'=\MaxStackHeight(P)$.
An analogous argument to the above concludes that we can remove one or both of them, to obtain a shorter positive cycle from $p$ to itself.

Since $|\Gamma|=2$, this means $|P|=2^{O(n^2)}$, which implies the statement of the lemma (witnessed by the worst case where $P$ decrements the counter in each step).
\end{proof}

\begin{proof}[Proof of \cref{lem:strongly_connected}]
Since $\P$ is bidirected, it suffices to argue that for every two states $p,q$, and $\sigma\in \StackAlphabetBot$, if we have a path $\tuple{p,\sigma}\DTo{} \tuple{p,v,\sigma}\DTo{} \tuple{q,\sigma}$ then we also have a path $\tuple{q,\sigma}\DTo{} \tuple{q,u,\sigma}\DTo{} \tuple{p,\sigma}$.
Observe that the condition for having the former path is that $\DFinite_{k-1}(u,v)=(a,\Some)$, for some $-\infty <a$, i.e., we have $\Paths{u}{v}_{k-1}\neq \emptyset$.
Since $\P$ is bidirected, we also have $\Paths{v}{u}_{k-1}\neq \emptyset$, and thus the latter path is also present in $G_k$.
\end{proof}

\begin{proof}[Proof of \cref{lem:1d_correctness}]
The statement holds clearly for $k=0$.
Now we argue that the statement also holds for arbitrary $k$, given that it holds for $k-1$.
In turn, it suffices to show that, for every two states $p,q$, for every path $P\in \Paths{p}{q}_{k-1}$ 
and $\sigma\in \StackAlphabetBot$,
we have a path $P'\in \Paths{\tuple{p,\sigma}}{\tuple{q,\sigma}}^{G_{k}}$ that is contained in $G_{k}\Project\sigma$ and such that $P\Dominated P'$.
This is sufficient because for any paths $P_1$, $P_2$, $P_1'$, $P'_2$ with $P_1\Dominated P'_1$ and $P_2\le P'_2$, we have $P_1\circ P_2\Dominated P'_1\circ P'_2$.

Let $\DFinite_{k-1}(p,q)=(a,b)$ and $\DOmega_{k-1}(p)=c$, and by the induction hypothesis we have $(a,b)=\DFinite^{G_{k-1}}(\tuple{p,\bot},\tuple{q,\bot})$ and $c=\DOmega^{G_{k-1}}(\tuple{p,\bot})$.
In particular, this implies that $\PathMin(P)\leq a$.
We distinguish the following cases.
\begin{itemize}
\item $b=\omega$.
Due to \cref{lem:summary_functions}, we have $c=a$.
We construct $P'$ by traversing the loop $\tuple{p, \sigma}\DTo{c} \tuple{p',\sigma}\DTo{-c+1}\tuple{p,\sigma}$ sufficiently many times, and then taking the edges $\tuple{p, \sigma}\DTo{a} \tuple{p,q,\sigma} \DTo{-a}\tuple{q,\sigma}$, to obtain $P\Dominated P'$.
\item $b<\omega$ and $\PathSum(P)\le b$. 
We construct $P'$ by traversing the edges $\tuple{p, \sigma}\DTo{a} \tuple{p,q,\sigma}\DTo{-a+b} \tuple{q,\sigma}$, to obtain $P\Dominated P'$.
\item $b<\omega$ and $\PathSum(P)>b$.
Due to \cref{lem:domega_pumping}, we have that $\PathMin(P)<c$.
We construct $P'$ by traversing the loop $\tuple{p, \sigma}\DTo{c} \tuple{p',\sigma}\DTo{-c+1}\tuple{p,\sigma}$ sufficiently many times, and then taking the edges $\tuple{p, \sigma}\DTo{a} \tuple{p,q,\sigma} \DTo{-a+b}\tuple{q,\sigma}$, to obtain $P\Dominated P'$.
\end{itemize}
\end{proof}

\begin{proof}[Proof of \cref{lem:1d_complexity}]
Due to \cref{lem:min_bounded}, \cref{lem:sum_bounded} and \cref{lem:min_cycle_bounded}, all elements of the summary functions $\DFinite_k$ and $\DOmega_k$ that are finite are bounded in absolute value by $2^{O(n^2)}$, and thus each iteration takes polynomial space.
Since for every $k>0$, the algorithm proceeds in the next iteration iff $\DFinite_{k-1}\StrictDominated \DFinite_{k}$ or $\DOmega_{k-1}\StrictDominated \DOmega_{k}$, we conclude that the algorithm terminates within $2^{O(n^2)}$ iterations.
\end{proof}

\subparagraph{Computing $\DFinite^G$ and $\DOmega^G$.}
Here we present an algorithm behind \cref{lem:dfinite_computation}.

Let $n=|V|$ be the number of nodes in $G$.
Given a simple cycle $C$ of $G$ with $\PathSum(C)>0$, we call a node $x$ in $C$ {\em critical} if the path $C_x$ obtained by traversing $C$ starting from $x$ is such that $\PathMin(C_x)=0$.
Note that every positive simple cycle has at least one critical node.
We assume without loss of generality that $G$ consists of a single (strongly) connected component, as otherwise we can repeat the computation on each component of $G$ independently.
Given a pair of values $(a,b)\in (\Z_{\leq 0}\cup\{-\infty\}) \times (\Z\cup\{-\infty \})$ and an integer $c\in \Z$, we define the extend operation $\Extend$ as 
\[
(a,b)\Extend c = ( \min(a, b+c), b+c ).
\]
The extend operation is monotonic, i.e.\ $(a,b) \le (a',b')$ implies $(a,b)\Extend c \le (a',b')\Extend c$.
The algorithm manipulates a graph data structure $G'$, initially identical to $G$, as follows.
\begin{enumerate}
\item\label{item:step1} We perform a standard Bellman-Ford computation on $G'$,
and detect whether there exists a positive cycle $C$.
If so, we identify a critical node $x$ of $C$.
We remove the edges outgoing $x$ from $G'$, and repeat this step.
\item\label{item:step2}
Let $X$ be the set of critical nodes identified in the previous step, 
and $G'=(V', E', \Weight')$ the graph data structure at the end of that step (note that $V'=V$).
For every node $u$, we compute $\DOmega^{G}(u)$ and $\DFinite^{G}(u,v)$ for all nodes $v\in V'$, as follows.
\begin{enumerate}
\item\label{step2a} 
We initialize a map data structure $\DS\colon V \to (\Z_{\leq 0}\cup\{-\infty\}) \times (\Z\cup\{-\infty \})$ with 
$\DS(u)=(0, 0)$ and $\DS(v)=(-\infty, -\infty)$ for all nodes $v\neq u$.
We then perform $n-1$ relaxation steps as follows.
In each step, we iterate over all edges $(y,v) \in E'$, and let $(a',b')=\DS(y)\Extend \Weight'(y,v)$.
If $\DS(v)\StrictDominated (a',b')$, we update $\DS(v)$ with $(a',b')$.
\item\label{step2b} 
We compute $c=\min(\{ a \colon x\in X \text{ and } \DS(x)=(a,\Some) \})$, and report that $\DOmega^{G}(u)=c$.
Given a node $v$, let $\DS(v)=(a_v,b_v)$ at the end of the previous step.
If $a_v\leq c$, we report that $\DFinite^G(u,v)=(c, \omega)$.
Otherwise, we report that $\DFinite^G(u,v)=\DS(v)$.
\end{enumerate}
\end{enumerate}

The algorithm clearly operates in polynomial time in the size of $G$, thus it remains to argue about the correctness.
We start with a straightforward lemma.

\begin{lemma}\label{lem:bf_correctness}
At the end of \cref{step2a} above, for every node $v$ we have $\DS(v)=\DFinite^{G'}(u,v)$.
\end{lemma}
\begin{proof}
For a natural number $i\in \N$, we let $\Paths{u}{v}^{G'}_i$ be the paths in $\Paths{u}{v}^{G'}$ of length $\leq i$.
and similarly denote by $\DFinite^{G'}_i$ the corresponding summary function restricted to paths in $\Paths{u}{v}^{G'}_i$.
We argue that the following assertions hold for all nodes $v$.
\begin{enumerate}
\item\label{item:inv1}
At all times, we have $\DS(v)\Dominated \DFinite^{G'}(u,v)$.
\item\label{item:inv2}
For all $i\in \{0,\dots, n-1\}$, at the end of iteration $i$, we have $\DFinite^{G'}_i(u,v) \Dominated \DS(v)$.
\end{enumerate}
Note that since $G'$ does not have positive cycles, \cref{item:inv2} implies that, in fact,  $\DFinite^{G'}(u,v) \Dominated \DS(v)$,
and hence the two items imply that $\DFinite^{G'}_i(u,v) = \DS(v)$.
We argue about each item separately.
\begin{enumerate}
\item The statement clearly holds after the initialization of $\DS$ with $\DS(u)=0$ and $\DS(v)=-\infty$ for all nodes $v\neq u$.
Now consider a point in which $\DS(v)$ is updated with $\DS(y)\Extend \Weight'(y,v)$, and assume that the statement holds for $\DS(y)$.
Thus we have a path $P\in \Paths{u}{y}^{G'}$ with $\DS(y)\Dominated (\PathMin(P), \PathSum(P))$.
We obtain a path $H\in \Paths{u}{v}^{G'}$ by extending $P$ with the edge $(y,v)$.
Observe that $\PathMin(H)=\min(\PathMin(P), \PathSum(P)+\Weight'(y,v))$ and $\PathSum(H)=\PathSum(P)+\Weight'(y,v)$.
Hence $(\PathMin(H), \PathSum(H))=(\PathMin(P), \PathSum(P)) \Extend \Weight'(y,v)$, and thus after the update
we have $\DS(v)\Dominated \DFinite^{G'}(u,v)$ by monotonicity of $\Extend$.
\item The statement clearly holds after the initialization of $\DS$, i.e., in iteration $0$.
Now assume it holds for some iteration $i-1$, for $i>0$, and we argue that it holds for iteration $i$.
Let $P$ be a path witnessing $\DFinite^{G'}_i(u,v)$.
If $|P|<i$, the statement holds by the induction hypothesis and the fact that $\DS(v)$ can only increase over the course of the algorithm.
Otherwise we have $|P|=i$, and $P$ occurs by extending a path $H$, having $|H|=i-1$, with an edge $(y,v)$.
By the induction hypothesis, we have $(\PathMin(H), \PathSum(H)) \Dominated \DS(y)$.
After the algorithm has processed the edge $(y,v)$, we have $(\PathMin(H), \PathSum(H))\Extend\Weight'(y,v)\Dominated \DS(v)$, and thus $(\PathMin(P), \PathSum(P))\Dominated\DS(v)$.
\end{enumerate}
The desired result follows.
\end{proof}

The following lemma yields the correctness of the algorithm.

\begin{lemma}\label{lem:finite_graphs_correctness}
Consider two nodes $u,v$ and let $\DFinite^G(u,v)=(a,b)$ and $\DOmega^G(u)=c$.
The following assertions hold.
\begin{enumerate}
\item $c=\min(\{a'\colon x\in X \text{ and }\DFinite^{G'}(u,x)=(a',\Some)\})$.
\item Either $a=c$ and $b=\omega$ or $a>c$ and $\DFinite^{G'}(u,v)=(a,b)$.
\end{enumerate}
\end{lemma}
\begin{proof}
We argue about each item separately.
\begin{enumerate}
\item Since $G'$ is a subgraph of $G$, we clearly have $c\geq \min(\{a'\colon x\in X \text{ and }\DFinite^{G'}(u,x)=(a',\Some)\})$,
witnessed by a path that reaches such a node $x$, traverses a positive cycle on which $x$ is critical sufficiently many times, and then returns to $u$.
We now argue that $c\leq \min(\{a'\colon x\in X \text{ and }\DFinite^{G'}(u,x)=(a',\Some)\})$.
Indeed, let $P$ be a path witnessing $\DOmega^G(u)$, and thus $\PathMin(P)=c$ and $\PathSum(P)>0$.
Since $P$ traverses a positive cycle and there are no positive cycles in $G'$, $P$ must contain a node $x\in X$,  thus $\PathMin(P)\leq \min(\{a'\colon x\in X \text{ and }\DFinite^{G'}(u,x)=(a',\Some)\})$, and hence $c\leq \min(\{a'\colon x\in X \text{ and }\DFinite^{G'}(u,x)=(a',\Some)\})$, as desired.
\item Due to \cref{lem:summary_functions} and the previous item, either $a=c$ and $b=\omega$, or $a>c$.
In the latter case there is a path $P$ witnessing $\DFinite^G(u,v)=(a,b)$ which does not contain any node in $X$.
Hence $P$ is a path in $G'$, and thus $\DFinite^{G'}(u,v)=(a,b)$.
\end{enumerate}
\end{proof}

\subparagraph{Exponential stack height}
We remark that even in 1 dimension, runs that witness reachability might require an exponential height on the stack.
As a simple example, consider a PVASS $\P$ with transitions 
(i)~$p\DTo{\sigma}p'$,
(ii)~$p'\DTo{1}p$
(iii)~$q\DTo{\sigma}q'$,
(iv)~$q'\DTo{1}q$,
as well as their bidirected variants (we use $\sigma$ to denote pushing on the stack).
Moreover, $p$ can reach $q$ via an exponentially long path (that only uses polynomial stack height) which decreases the counter in almost every step in the first half of the path, and increases the counter in almost every step in the second half.
This effect can be easily generated by serially connecting two similar gadgets: each gadget uses the stack as a binary counter, which, before it ticks, it decrements (for the first gadget) and increments (for the second gadget) the PVASS counter
(we may simply make the two gadgets use different stack symbols to easily discard paths that might alternate between the two due to bidirectedness).
Note that $q$ is reachable from $p$, but any witness run requires to initially traverse the $p\DTo{\sigma}p'\DTo{1}p$ cycle exponentially many times in order to cross the exponentially deep valley generated by the two gadgets.
In effect, this results in increasing the stack height by an exponential amount.

%% file: app-tower.tex
\directedgadget*

\begin{proof}
We say that a gadget $G$ is {\em $i$-neutral} if $(s,\bm{u},w) \xrightarrow{*}_G (t,\bm{v},w')$
implies that $\bm{u}$ is $i$-initialized if and only if $\bm{v}$ is $i$-initialized.
Observe that the statement implies that $G_k$ is $k$-neutral
since both $[0, \dots, 0]$ and $[0, \dots, 0, \exp^k(n)]$ are $k$-initialized.

We proceed by induction over $k$ where the case $k = 1$ is clear and we assume $k > 1$.
It is easy to verify the existence of the runs $(q_0,[0, \dots, 0],w) \xrightarrow{*}_{G_k} (q_5,[0, \dots, 0, \exp^k(n)],w)$,
which traverses the gadget Inc$_k$ exactly $\exp^k(n)-1$ times.
It remains to verify that these are the only $k$-initialized runs from $q_0$ to $q_5$.
To this end, we decompose such a run into subruns, along the several subgadgets it traverses.
By arguing that all subgadgets are $(k-1)$-neutral we can conclude that all subruns must be $(k-1)$-initialized
so that we can apply the induction hypothesis to the subruns.

First consider the gadgets $Z_{k-1}$ and Inc$_k$.
The only $(k-1)$-initialized runs in $Z_{k-1}$ between its terminal states, say $s$ and $t$,
are the runs of the form
\[
	(s,[0, \dots, 0],w) \xrightarrow{*}_{Z_{k-1}} (t,[0, \dots, 0],w) \quad \text{for } w \in \Gamma_{k-1}^*,
\]
since $G_{k-1}$ satisfies the induction hypothesis and is $(k-1)$-neutral.
Moreover, we claim that the only $(k-1)$-initialized runs in Inc$_k$ from $p_0$ to $p_4$ are the runs of the form
\begin{equation}
	\label{incrun}
	(p_0,[0, \dots, 0, x_k],w \texttt{0}\texttt{1}^\ell) \xrightarrow{*}_{\text{Inc}_k} (p_4,[0, \dots, 0, x_k + 1],w \texttt{1}\texttt{0}^\ell)
\end{equation}
for $w \in \Gamma_{k-1}^*$ and $\ell \le \exp^{k-1}(n)$.
This follows from the previous statement for $Z_{k-1}$, the fact that the stack determines how often the $p_1$-loop is executed
and the fact that the $p_2$-loop must be repeated until $c_{k-1} = 0$.

Now consider a $k$-initialized run in $G_k$ from $q_0$ to $q_5$.
It can be decomposed into subruns which traverse the gadgets $G_{k-1}$, $\bar G_{k-1}$ and $Z_{k-1}$,
and subruns that do not manipulate the counters $c_j, \bar c_j$ with $j < k-1$.
In particular, all these subruns are $(k-1)$-initialized, which allows us to apply the uniqueness properties for $G_{k-1}$ and $Z_{k-1}$:
The $Z_{k-1}$-runs must start and end with a vector $[0, \dots, 0, x_k]$
whereas the $G_{k-1}$-runs must start with a vector $[0, \dots, 0, x_k]$ and end with $[0, \dots, \exp^{k-1}(n), x_k]$.
This in turn implies that the loops on $q_1$ and $q_3$ must be executed exactly $\exp^{k-1}(n)$ times.
The gadget Inc$_k$ must be executed exactly $\exp^k(n)-1$ times, considering that the binary counter on the stack
must be incremented from $\texttt{0}^{\exp^{k-1}(n)}$ to $\texttt{1}^{\exp^{k-1}(n)}$ according to \eqref{incrun}.
\end{proof}

\bidirectedgadget*

\begin{proof}

Again we proceed by induction over $k$ where $k=1$ is clear.
For $k>1$ consider a $\leftrightarrow_{G_k}$-derivation
witnessing $(s,\bm{u},w)  \stackrel{*}{\leftrightarrow}_{G_k} (t,\bm{v},w')$,
with the minimal number of $\leftarrow_{G_k}$-steps.
Whenever the derivation passes through a subgadget $G_{k-1}$ or $\bar G_{k-1}$,
this part can be replaced by a forward derivation by induction hypothesis.
Moreover, the $\leftrightarrow_{G_k}$-derivation cannot contain subruns which enter and exit a subgadget $G_{k-1}$ or $\bar G_{k-1}$
through the same terminal state since, again by induction hypothesis, such a subrun could be cut out,
reducing the number of reverse transitions.
Hence, the $\leftrightarrow_{G_k}$-derivation can only contain reverse transitions which are not contained
in a subgadget $G_{k-1}$ or $\bar G_{k-1}$.
Consider the first occurrence of a reverse transition $\bar \tau$.
By a case distinction, one can argue that $\bar \tau$ must always be preceded by $\tau$, and hence we can cancel $\tau$ and $\bar \tau$.
If $\tau$ is the loop on $q_1$ then the occurrence of $\bar \tau$ must be preceded by either $\tau$ itself,
in which case $\tau$ and $\bar \tau$ can be cancelled, or a run through $G_{k-1}$.
However, after executing $G_{k-1}$ the counter $\bar c_{k-1}$ must be zero and $\bar \tau$ cannot decrement $\bar c_{k-1}$.
Similar arguments hold for the loops on $q_3$ and $p_1$, and the transitions $p_3 \to p_4$ and $q_3 \to q_4$.
If $\tau$ is the loop on $p_2$ then $\bar \tau$ is either preceded by $\tau$ or by the transition $p_1 \to p_2$.
However, since the transition $p_1 \to p_2$ pushes $\texttt{1}$ the symbol $\texttt{0}$ cannot be popped by $\bar \tau$.

In a very similar fashion we can prove the second statement,
by considering a $\leftrightarrow_{G_k}$-cycle with the minimal number of backward transitions.
Using the induction hypothesis we can cut out cycles on the gadgets $G_{k-1}$ and $\bar G_{k-1}$,
and replace bidirected derivations passing through $G_{k-1}$ or $\bar G_{k-1}$ by forward derivations.
It remains to deal with backward transitions, which can be done with the same arguments as above.
\end{proof}

%% file: app-pspace-lower-bound-failure.tex
The paper~\cite{EnglertHLLLS21} presents a $\PSPACE$-hardness
proof for the coverability problem in $1$-PVASS that could be considered a
candidate to show $\PSPACE$-hardness of coverability (or just reachability) in
bidirected $1$-PVASS. However, it fails for non-obvious reasons. Here, we will
briefly outline why it fails.  More specifically, we argue that simply making
the $1$-PVASS from that proof bidirected (by adding backwards edges) will not
yield a correct reduction.  We assume familiarity with the construction from
\cite{EnglertHLLLS21}.

The reduction in~\cite{EnglertHLLLS21} reduces from the alternating subset sum problem. Here, we are given numbers $a_1,a'_1,\ldots,a_k,a'_k$, $e_1,e'_1,\ldots,e_k,e'_k$, and $s$ in $\N$ (in binary) and we want to decide whether
\[ \forall x_1\in\{a_1,a'_1\}~\exists y_1\in\{e_1,e'_1\}\cdots\forall x_k\in\{a_k,a'_k\}~\exists y_k\in\{e_k,e'_k\}\colon~x_1+y_1+\cdots+x_k+y_k=s. \]
This can be seen as a two-player game with a universal player that picks each
$x_i\in\{a_i,a'_i\}$ and an existential player that picks each
$y_i\in\{e_i,e'_i\}$. The instance is positive if the existential player has a
winning strategy, which means that it ensures $x_1+y_1+\cdots+x_k+y_k=s$.

We say that a sequence of transitions in a PVASS is a \emph{pseudo-run} if the
counter is allowed to go below zero.  The 1-PVASS constructed
in~\cite{EnglertHLLLS21} has the property that each
pseudo-run $(q_0,\varepsilon,0)\to(q_f,\varepsilon,m)$ corresponds to a
strategy for the existential player. Moreover, it is a run if it is indeed a
winning strategy for the existential player.  

Now suppose that (i)~every play will overshoot $s$, i.e.\
 $x_1+y_1+\cdots+x_k+y_k>s$ (in particular, there is no winning strategy), but
(ii)~there are two strategies for $E$ that will, with the same moves by $A$, overshoot
in distinct ways. Specifically, there are $x_i\in\{a_i,a'_i\}$ for $i\in[1,k]$
and $y_i\in\{e_i,e'_i\}$ for $i\in[1,k-1]$ such that
\begin{align*}
x_1+y_1+\cdots +x_k+e_k&>s \\
x_1+y_1+\cdots +x_k+e'_k&>s
\end{align*}
and $e_k\ne e'_k$. Then there are two runs in the constructed $1$-PVASS of the form
\begin{align}
(q_0,\varepsilon,0)\xrightarrow{*}(p,w,m)\xrightarrow{*} (q,w,m-\delta_1)~\text{and}~(p,w,m)\xrightarrow{*}(q,w,m-\delta_2), \label{pspace-failure-two-runs}
\end{align}
where $\delta_1\ne\delta_2$. This is because at the top of the stack (which
then has the form $wv$ for some $v$), it is checked that each sum is at least
$s$. However, if the sum is more than $s$, this only leads to issues all the
way at the end of the run, because going above $s$ will remove too many tokens
to sustain a run to $(q_f,\varepsilon,0)$. Therefore, in the $1$-PVASS,
there is no run in this case.

However, if we add backwards edges, the two runs in
\eqref{pspace-failure-two-runs} allow us to create an arbitrary supply of
tokens: Suppose $\delta_1<\delta_2$ and set $\delta=\delta_2-\delta_1>0$. Then
we can execute:
\begin{align}
(p,w,m)\xrightarrow{*}(q,w,m-\delta_1)\xrightarrow{*}(p,w,m-\delta_1+\delta_2)=(p,w,m+\delta).
\end{align}
And thus $(p,w,m)\xrightarrow{*}(p,w,m+\ell\delta)$ for every $\ell\in\N$.
However, in the constructed $1$-PVASS, it is easy to see that for every
reachable configuration $(r,v,n)$, there exists some $n'$ such that from
$(r,v,n')$, we can cover $(q_f,\varepsilon,0)$, this shows that we can cover
$(q_f,\varepsilon,0)$ even if there is no winning strategy.